\documentclass[reqno]{amsart}
\usepackage[margin=1.5in,bottom=1.25in]{geometry}		% for tighter margins (80 CPL)

\usepackage{amsmath}		% for AMS macros
\usepackage{amssymb}		% for AMS symbols
\usepackage{amsfonts}		% for AMS fonts
\usepackage{amsthm}		% for theorems
\usepackage[foot]{amsaddr}		% for author footnotes

\usepackage{mathtools}		% for advanced math

\mathtoolsset{%
}

\usepackage[utf8]{inputenc}		% for source encoding
\usepackage[T1]{fontenc}		% for font encoding

\usepackage[%		% for math font selection
cal=cm,
]
{mathalfa}

\usepackage{dsfont}		% for blackboard bold font
\usepackage[sf,mono=false]{libertine}
\usepackage[labelfont={bf,small},labelsep=colon,font=small]{caption}	% for caption control
\usepackage[dvipsnames,svgnames]{xcolor}		% for color
\colorlet{MyRed}{DarkViolet}
\colorlet{MyGreen}{DarkGreen!80!Black}
\colorlet{MyBlue}{MediumBlue}

\usepackage{paralist}
\usepackage[misc]{ifsym}

\usepackage[sort&compress]{natbib}		% for citations

\bibpunct[, ]{[}{]}{,}{n}{,}{,}

\usepackage{hyperref}
\hypersetup{
colorlinks=true,
linktocpage=true,
pdfstartview=FitH,
breaklinks=true,
pdfpagemode=UseNone,
pageanchor=true,
pdfpagemode=UseOutlines,
plainpages=false,
bookmarksnumbered,
bookmarksopen=false,
bookmarksopenlevel=1,
hypertexnames=false,
pdfhighlight=/O,
urlcolor=MyBlue,linkcolor=MyBlue,citecolor=MyBlue,	% for on-screen
pdftitle={},
pdfauthor={},
pdfsubject={},
pdfkeywords={},
pdfcreator={pdfLaTeX},
pdfproducer={LaTeX with hyperref}
}

\newcommand{\EMAIL}[1]{\email{\href{mailto:#1}{#1}}}

\usepackage[sort&compress,capitalize,nameinlink]{cleveref}		% for cleveref formatting
\crefname{assumption}{Assumption}{Assumptions}
\crefname{claim}{Claim}{Claims}
\theoremstyle{plain}
\newtheorem{theorem}{Theorem}		% for theorems
\newtheorem{proposition}{Proposition}		% for propositions

\newtheorem{claim}{Claim}		% for claims

\newtheorem*{corollary*}{Corollary}		% for corollaries (unnumbered)

\theoremstyle{definition}
\newtheorem{definition}{Definition}		% for definitions
\newtheorem{example}{Example}		% for examples
\newtheorem*{definition*}{Definition}		% for definitions (unnumbered)
\newtheorem*{assumption*}{Assumptions}		% for assumptions (unnumbered)
\newtheorem*{example*}{Example}		% for examples (unnumbered)

\theoremstyle{remark}
\newtheorem*{remark*}{Remark}		% for remarks (unnumbered)

\newcommand{\ep}{\varepsilon}

\newcommand{\doititle}[1]{#1}

\begin{document}

%**********************************************************************
%***    FRONT MATTER AND METADATA
%**********************************************************************

%----------------------------------------------------------------------
%%% TITLE & AUTHORS
%----------------------------------------------------------------------
\title
[Survival of dominated strategies]
{Survival of dominated strategies under imitation dynamics}

%-------------------------------------------------------------------
\author
[P.~Mertikopoulos]
{Panayotis Mertikopoulos$^{\ast}$}
\address{$^{\ast}$\,%
Univ. Grenoble Alpes, CNRS, Inria, Grenoble INP, LIG, 38000 Grenoble, France}
\EMAIL{panayotis.mertikopoulos@imag.fr}
%-------------------------------------------------------------------
\author
[Y.~Viossat]
{Yannick Viossat$^{\diamond,\sharp}$}
\address{$\diamond$\,%
CEREMADE, Université Paris Dauphine-PSL, Place du Maréchal de Lattre de Tassigny, F-75775 Paris, France
}
\address{$\sharp$\,%
Corresponding author}
\EMAIL{viossat@ceremade.dauphine.fr}
%-------------------------------------------------------------------

%----------------------------------------------------------------------
%%% KEYWORDS
%----------------------------------------------------------------------


%----------------------------------------------------------------------
%%% THANKS
%----------------------------------------------------------------------
\thanks{This article is dedicated to the memory of Bill Sandholm, who, had he lived, would have been a co-author of this work.
We thank him, Vianney Perchet, Jorge Pe\~{n}a, seminar audiences, and two anonymous reviewers for helpful comments.}

%----------------------------------------------------------------------
%%% ABSTRACT
%----------------------------------------------------------------------
\begin{abstract}
The literature on evolutionary game theory suggests that pure strategies that are strictly dominated by other pure strategies always become extinct under imitative game dynamics,
but they can survive under innovative dynamics. As we explain, this is because innovative dynamics favour rare strategies while standard imitative dynamics do not. However, as we
also show, there are reasonable imitation protocols that favour rare or frequent strategies, thus allowing strictly dominated strategies to survive in large classes of imitation dynamics.
Dominated strategies can persist at nontrivial frequencies even when the level of domination is not small.
\end{abstract}

%**********************************************************************
%***    BODY TEXT
%**********************************************************************
\maketitle
\allowdisplaybreaks		% for breaking long displays

%----------------------------------------------------------------------
%%% INTRODUCTION
%----------------------------------------------------------------------
\section{Introduction}
\label{sec:intro}

Many economic models assume that the agents they consider are rational.
This may be defended as a reference case or for tractability.
A more interesting justification is that, at least in tasks that they perform routinely, and for which they have enough time to experiment, even weakly rational agents should come to learn which strategies do well, and behave eventually \emph{as if} they were rational.
The same intuition applies to other evolutionary processes, such as natural selection or imitation of successful agents. But does evolution really wipe out irrational behaviors?

A simple way to tackle this question in a game-theoretic context is to study whether evolutionary game dynamics wipe out dominated strategies, in the sense that the frequency of these strategies goes to zero as time goes to infinity. This may be interpreted in several ways, depending on whether domination means weak or strict domination, whether the strategies considered are pure or mixed, and the dynamics deterministic or stochastic (see Viossat, 2015 \cite{11}, for a partial survey).
We focus here on what we see as the most basic question:
\emph{do pure strategies that are strictly dominated by other pure strategies become extinct under deterministic dynamics in continuous time?}

The answer of the literature is mixed.
Roughly speaking, evolutionary game dynamics may be classified as imitative or innovative.
In imitative dynamics, modeling imitation processes or pure selection (without mutation),
strategies that are initially absent from the population never appear.
The leading example is the replicator dynamics.
In innovative dynamics, strategies initially absent from the population may appear.
Examples include the best-reply dynamics (and smoothened versions of it), the Brown-von Neumann-Nash dynamics, the Smith dynamics, the projection dynamics, and others.

The literature shows that imitative dynamics (in the sense of Sandholm, 2010 \cite{10}) always eliminate pure strategies strictly dominated by other pure strategies (Akin, 1980 \cite{1}; Nachbar, 1990 \cite{7}), while innovative dynamics need not do so, with the notable exception of the best-reply dynamics.
Indeed, building on Berger and Hofbauer (2006) \cite{2}, Hofbauer and Sandholm (2011) \cite{5} show that for all dynamics satisfying four natural conditions called Innovation, Continuity, Nash Stationarity and Positive Correlation, there are games in which pure strategies strictly dominated by other pure strategies survive in relatively high proportion.
Moreover, their simulations show that, at least for some well-known dynamics, dominated strategies may survive at non-negligible frequencies even
when the difference in payoff between the dominated and dominating strategies is relatively important.
Thus, with respect to elimination of dominated strategies, there seems to be a sharp contrast between imitative and innovative processes.

This paper argues that this is not the case.
As we shall explain, the intuitive reason why innovative dynamics allow for survival of dominated strategies is that they give an edge to rare strategies.
Indeed, the \emph{Innovation} property of Hofbauer and Sandholm stipulates that if a strategy is an unplayed best-response to the current population state, then it should appear in the population: technically, the derivative of its frequency should be positive.
The per-capita growth rate of its frequency is then infinite.
Moreover, the \emph{Continuity} property requires that the dynamics depends smoothly on the payoffs of the game and the population state.
Taken together, these two properties imply that rare strategies that are almost-best replies to the current population state have a huge per-capita growth rate, potentially higher than strategies that have a slightly better payoff, but are more frequent.
In this sense, Hofbauer and Sandholm's dynamics favour rare strategies.
When a dominated strategy becomes rare, this advantage to rarity may compensate the fact of being dominated and allows it to survive.

By contrast, in imitative dynamics, the per-capita growth rates of pure strategies are always ordered as their payoffs,
irrespective of their frequencies in the population. But we feel that this is, in some sense, an artifact, a legacy of the history of evolutionary game theory.
Indeed, imitative dynamics arose as variants of the replicator dynamics, which originated as a natural selection model, and was only a posteriori reinterpreted as an imitation model. Ironically, their rationality properties come from their biological interpretation.
But if we consider a priori which dynamics could arise from an imitation protocol, then we arrive quite naturally at dynamics that provide an evolutionary advantage to rare strategies (or frequent strategies) in a sense that we will make clear. As in innovative dynamics, this advantage to rarity (or commonness) may offset the fact of being dominated, hence allowing dominated strategies to survive.

More precisely, imitative dynamics may be derived through a two-step imitation protocol. In the first step, an agent (henceforth, the \emph{revising agent}) meets another individual (the \emph{mentor}) uniformly at random.
In an infinite population, the probability that the mentor plays a given strategy is thus equal to the frequency of this strategy. In the second step, the revising agent decides to adopt the mentor's strategy or to
keep his own. The adoption rule depends on the dynamics but satisfies a monotonicity condition. Roughly, the probability of switching is larger if the revising agent's payoff is low, the mentor's payoff is large,
or both. This leads to dynamics that coincide with Nachbar's (1990) \cite{7} monotone dynamics: if strategy $i$ has a larger current payoff than strategy $j$, then its frequency has a larger per-capita growth-rate.
We thus suggest to call them \emph{monotone imitative dynamics}.\footnote{We thank an anonymous reviewer for suggesting this name.}

To motivate more general, non-monotone imitative dynamics, we consider revision protocols where the second step satisfies the standard monotonicity condition, but the first step is modified. Instead of always meeting a single other individual, a revising agent sometimes meets several. There are then many reasonable ways of choosing a mentor (or depending on the interpretation, a strategy to be potentially imitated). The probability of envisioning to switch to a given strategy may then be lower or higher than the frequency of this strategy, in a way that may systematically favour rare or frequent strategies. This leads to dynamics that are no longer monotone in the sense of Nachbar (1990)  \cite{7}, and under which dominated strategies may survive. Jorge Pen\~a brought to our attention that similar phenomena have been studied in the literature on the evolution of cooperation. In particular, a conformist bias may allow cooperation to survive in the prisoner's dilemma (e.g., Boyd and Richerson, 1988 \cite{3}; Heinrich and Boyd, 2001 \cite{4}; Pe\~na et al., 2009 \cite{8}; and references therein).

We first illustrate these ideas on dynamics derived from imitation protocols based on adoption of successful strategies or departure from less successful ones,
but not on direct comparison between the payoff of an agent's current strategy and of the strategy he envisions to adopt. With such protocols, agents keep switching
from a strategy to another even when all strategies earn the same payoff.
For this reason, an advantage to rare or frequent strategies always bites, and dominated strategies may survive even in games with only two strategies.
The argument is simple: if the two strategies are twins, that is, always earn the same payoffs, then in the case of an advantage to rare strategies, the shares of both strategies tend to become equal. Technically, the population state where both strategies are played with probability 1/2 is globally asymptotically stable. If we penalize sufficiently little one of the strategies, to make it dominated, most solutions still converge to one or several rest points in the neighborhood of this population state, in which the dominated strategy is played with positive probability.

Of course, these rest points cannot be Nash equilibria. This reveals that the dynamics we just mentioned do not satisfy the evolutionary folk theorem (see, e.g., Weibull, 1995 \cite{12}). They do not satisfy either the Positive Correlation condition, which stipulates that there is a positive correlation between the growth rates of strategies and their payoffs (or, equivalently, that against a constant environment, the average payoff in the population increases). Our main result is to show that survival of dominated strategies also occurs under dynamics that are derived from imitation protocols based on payoff comparison,
and that satisfy both the evolutionary folk theorem and an appropriate version of Positive Correlation. We show that this is the case as soon as they also satisfy the Continuity condition of Hofbauer and Sandholm and two additional conditions: Imitation, and Advantage to Rarity. The former requires that, except at Nash equilibria, a strategy which is currently played must be abandoned by some agents or imitated by others (or both). The latter assumes that if two strategies are twins, then the rarer one has a per-capita growth-rate that is no lower than the per-capita growth-rate of the more frequent one, and strictly higher in some precise circumstances. The Advantage to Rarity condition may be replaced by a similar Advantage to Frequency. We provide a number of imitation protocols leading to dynamics satisfying these assumptions.

Under these dynamics, if a solution converges to a rest point, this point must be a Nash equilibrium, hence put a zero weight on all strictly dominated strategies.
Therefore, to prove that dominated strategies may survive, we need to consider games where solutions cycle. We consider the same game as Hofbauer and
Sandholm, the hypnodisk game with a feeble twin, and use similar arguments, with some twists.
We check via simulations that dominated strategies can also survive in more standard games, such as a Rock-Paper-Scissors game augmented by a feeble twin
of Scissors, as also considered by Hofbauer and Sandholm. Finally, we show that simpler examples of survival of dominated strategies can be given if we depart
from single population dynamics and consider a population of agents facing an environment which oscillates for exogeneous reasons.

The remainder of this article is organized as follows. Evolutionary dynamics are introduced in \cref{sec:EvolDyn}. \cref{sec:ImProc} describes
imitation processes favouring rare strategies or frequent strategies. \cref{sec:simple} gives a simple example of survival of dominated strategies under
dynamics based on protocols known as the imitation of success, or imitation driven by dissatisfaction. \cref{sec:paycomp} states our main results: that
survival of dominated strategies also occurs for imitation dynamics based on payoff comparison, and for any imitation dynamics satisfying some natural conditions,
on top of favouring rare or frequent strategies. The result is proved in \cref{sec:proof}. \cref{sec:disc} concludes. \cref{app:proofs} gathers
some proofs. \cref{app:moreprot} discusses more general imitation protocols than those described in the main text.
Finally, \cref{app:unilateral} gives simple examples of survival of dominated strategies under dynamics based on payoff comparison in
a population playing against an ad-hoc environment.

%----------------------------------------------------------------------
%%% DYNAMICS
%----------------------------------------------------------------------
\section{Evolutionary dynamics}
\label{sec:EvolDyn}

With the exception of \cref{app:unilateral}, we focus on single-population dynamics.
There is a single, unit mass population of agents.
These agents may choose any pure strategy in the set $I = \{1, \dotsc, N\}$.
The frequency of strategy $i$ at time $t$ is denoted by $x_i(t)$.
The vector $x(t) = (x_i(t))_{i \in I}$ of these frequencies is called the population state at time $t$.
It belongs to the simplex $X = \{x \in \mathbb{R}^N_+, \sum_{i \in I} x_i = 1\}$.
The payoff for an agent playing strategy $i$ when the population state is $x$ is denoted by $F_i(x)$.
The vector $F(\cdot)=(F_1(\cdot),\dotsc,F_N(\cdot)) : X \to \mathbb{R}^N$
is called the game's payoff function. We frequently identify a (symmetric two-player) game and its payoff function.

We are interested in evolutionary dynamics of the form $\dot{x}= V^F(x)$, with $V^F$ Lipschitz continuous in $x$, to ensure existence and uniqueness of solutions through a given initial condition. Thus, the population state evolves as a function of the current state and the payoffs of the game. The vector field $V^F$ is assumed to depend continuously on the game's payoff function $F$.\footnote{To fix ideas, we use the sup norm on the space of payoff functions: $||F|| = \sup_{x \in X, i \in I} |F_i(x)|$, and again the sup norm $||(F,x)|| = \max(||F||, ||x||)$ to define joint continuity in $(F, x)$. This is not essential.}

A well-known example is the replicator dynamics:
\begin{equation}
\label{eq:rep}
\dot x_i(t) = x_i(t)  \left[F_i(x(t)) - \bar{F}(x(t))\right]
\end{equation}
where $\bar{F}(x(t))= \sum_{i \in I} x_i(t) F_i(x(t))$ is the average payoff in the population.
We often omit to specify that the payoffs depend on the state, which depends on time.
Thus, instead of \eqref{eq:rep}, we write: $\dot x_i = x_i (F_i - \bar{F})$.

Pure strategy $i$ is strictly dominated by pure strategy $j$ if for all $x$ in $X$, $F_i(x) < F_j(x)$.
Pure strategy $i$ goes extinct, along a given solution of given dynamics, if $x_i(t) \to 0$ as $t \to +\infty$.
We want to understand under which dynamics pure strategies strictly dominated by other pure strategies always go extinct, at least for initial conditions in which all strategies are initially present, that is, in the relative interior of the simplex $X$.

Before introducing imitative and innovative dynamics, let us explain a standard way to derive dynamics from micro-foundations.
The idea is that from time to time agents revise their strategies.
Due to this revision process, agents playing strategy $i$ switch to strategy $j$ at a certain rate, which depends on the population state and on the payoffs of the game.
We denote this rate by $\rho_{ij}(x, F)$, or simply $\rho_{ij}$ to keep formulas light. Thus, between time $t$ and $t+ dt$, a mass $x_i \rho_{ij} dt$ of agents switch from $i$ to $j$, and a mass $x_j \rho_{ji} dt$ switch from $j$ to $i$.
This leads to the ``mother equation":
\begin{equation}
\label{eq:mother}
\dot x_i  = \sum_{j \neq i} x_j \rho_{ji} - x_i \sum_{j \neq i} \rho_{ij}
\end{equation}
where the first term is an inflow term (agents starting to play strategy $i$) and the second term an outflow
term (agents abandoning strategy $i$).\footnote{As the terms $i=j$ cancel, Eq. \eqref{eq:mother} may also be written as follows:
\[\dot x_i  = \sum_{j \in I} x_j \rho_{ji} - x_i \sum_{j \in I} \rho_{ij} \]}

A specification of the rates $\rho_{ij}$ for all $(i, j)$ in $I \times I$ is called a revision protocol and defines dynamics.
The replicator dynamics for instance may be derived from at least three different protocols.

\begin{itemize}
\item (imitation of success) $\rho_{ij} = x_j (K + F_j(x))$, where $K$ is a constant large enough to ensure that $K+ F_j(x)$ is positive for all strategies $j$ in $I$ and all states $x$ in $X$.

\item (imitation driven by dissatisfaction) $\rho_{ij} = x_j (K - F_i(x))$, with $K > F_i(x)$ for all $i$ in $I$ and all $x$ in $X$.

\item (proportional pairwise imitation rule) $\rho_{ij} = x_j [F_j- F_i]_+$, where for any real number $a$, $[a]_+=\max(a, 0)$.

\end{itemize}

These three protocols model two-step processes: first, a revising agent meets another agent uniformly at random, hence playing $j$ with probability $x_j$;
second, he imitates her with a probability that depends on the payoff of this agent's strategy, his own, or a comparison of both.\footnote{We use ``He" for the revising agent, and ``She" for the agent being imitated.}

\emph{Imitative dynamics.} More generally, Sandholm (2010) \cite{10} calls dynamics imitative if they may be derived from a revision protocol of the form  \[\rho_{ij} (F,x) = x_j r_{ij}(F, x)\]
with for all $x$ in $X$, all strategies $i, j, k$ in $I$:
\begin{equation}
\label{eq:monotonicity}
F_i(x) < F_j(x) \Leftrightarrow r_{kj}(F,x) - r_{jk}(F,x) > r_{ki}(F, x) - r_{ik}(F,x)
\end{equation}

As the replicator dynamics, these dynamics may be seen as modeling a two-step process where, in step 1, a revising agent meets another agent from the population at random, and in step 2, decides to imitate her or not.
Condition \eqref{eq:monotonicity} is a monotonicity condition.
It means that in step 2, the difference between the conditional imitation rates from $k$ to $i$ and from $i$ to $k$ increases with the payoff of strategy $i$. In particular, if strategy $j$ earns more than strategy $i$, then in step 2, an agent playing strategy $i$ is more likely to adopt $j$ than an agent playing $j$ is to adopt $i$.

It is easy to see that imitative dynamics coincide with a class of dynamics known as monotone dynamics (Viossat, 2015 \cite{11}, footnote 6).
These are dynamics of the form
\[\dot{x}_i=x_i g_i(x)\]
with $g_i$ Lipschitz continuous and, for all $x \in X$, and all $(i, j)$ in $I \times I$,
\begin{equation*}
%\label{eq:mondyn2}
g_{i}(x) < g_{j}(x) \Leftrightarrow F_{i}(x)  <  F_{j}(x).
\end{equation*}
It follows that in imitative dynamics, per-capita growth rates of pure strategies are ordered as their payoffs.
As a result, pure strategies strictly dominated by other pure strategies are always eliminated (Akin, 1980 \cite{1}; Nachbar, 1990 \cite{7}; Samuelson and Zhang, 1992 \cite{9}; Hofbauer and Weibull, 1996 \cite{5}).
To distinguish them from more general imitation processes that we will consider, we refer to these dynamics as \emph{monotone imitative dynamics}.  This monotone character does not only derive from the monotonicity condition \eqref{eq:monotonicity}, but also from the assumption that the probability of envisioning to adopt a given strategy is equal to the frequency of this strategy.

\emph{Innovative dynamics.} By contrast with imitative dynamics, in innovative dynamics, strategies that are not initially played may appear.
A leading example is the Smith dynamic:
\begin{equation}
\label{eq:Smith}
\dot{x_i} = \sum_{j \in I} x_j [F_i(x) - F_j(x)]_+ - x_i \sum_{i \in I} [F_j(x) - F_i(x)]_+
\end{equation}
It may be derived by assuming that, first, revising $i$-strategists\footnote{An $i$-strategist is an agent currently using strategy $i$.} pick a strategy $j$ uniformly at random in the list of possible strategies, and second, adopt it with probability proportional to $[F_j - F_i]_+$.
 This leads to $\rho_{ij} = \frac{1}{N} [F_j - F_i]_+$.
This is similar to the proportional pairwise imitation rule defining the replicator dynamics, except that in the first step, strategy $j$ is selected as a candidate new strategy with probability $1/N$ instead of $x_j$.\footnote{In Eq.\eqref{eq:Smith}, as standard, we omitted the factor $1/N$, which only affects the time-scale.}

Other well known innovative dynamics are the Brown-von Neumann-Nash dynamics, or BNN:
\begin{equation*}
%\label{eq:BNN}
\dot x_i = \left[F_i(x) - \bar{F}(x)\right]_+ - x_i \sum_{k \in I} [F_k(x) - \bar{F}(x)]_+
\end{equation*}
They model a two-step process where, in step 1, revising $i$-strategists pick a strategy $j$ uniformly at random in the list of possible strategies, and, in step 2, adopt it with probability proportional to $[F_j - \bar{F}]_+$, where $\bar{F}(x)=\sum_i x_i F_i(x)$ is the average payoff in the population.

\emph{Innovative Dynamics favour rare strategies, monotone imitative dynamics do not.} Building on Berger and Hofbauer (2006) \cite{2}, Hofbauer and Sandholm (2011) \cite{5} showed that for the Smith and BNN dynamics, and many others, there are games in which a pure strategy strictly dominated
by another pure strategy survives, for most initial conditions. This holds for any dynamics satisfying four natural requirements, called \emph{Innovation},
\emph{Continuity}, \emph{Positive Correlation} and \emph{Nash Stationarity}. As explained in the introduction, the intuition is that, taken together, Innovation and Continuity favour rare strategies, in the sense that a rare strategy can have a higher per-capita growth-rate than a better but more frequent strategy.

By contrast, monotone imitative dynamics favour neither rare nor frequent strategies: they are neutral. Under monotone imitative dynamics,
if the payoff of strategy $i$ is less than the payoff of strategy $j$, then its per-capita growth rate is less than that of strategy $j$. This is true whatever the frequencies of strategies $i$ and $j$. The reason is not that this property is completely natural. Indeed, it does not hold for innovative dynamics. Rather, this is because the imitation processes modeled
by monotone imitative dynamics are of a particular kind, inspired by the replicator dynamics. It is actually easy to imagine dynamics modeling imitation processes
but advantaging rare strategies, or frequent ones.\footnote{Of course, such dynamics, though modeling imitation processes, do not satisfy Sandholm's definition of
imitative dynamics. This is the key-point: this definition of imitative dynamics does not encompass all reasonable imitation processes.}  For these dynamics, as for
innovative dynamics, the advantage given to rare (or frequent) strategies should be able to offset the fact of being strictly dominated, allowing for survival of dominated strategies.
This is what we show.
We begin by providing examples of imitation dynamics favouring rare or frequent strategies. They are all based on the idea that instead of deciding to change his strategy or not upon meeting only one other agent, a revising agent might meet several other agents before taking his decision.

%----------------------------------------------------------------------
%%% IMITATION
%----------------------------------------------------------------------
\section{Imitation processes advantaging rare or frequent strategies}
\label{sec:ImProc}

\subsection{Examples} Loosely speaking, dynamics favour rare strategies if, when strategies $i$ and $j$ earn the same payoff but strategy $i$ is rarer, strategy $i$ has a higher per-capita growth rate than strategy $j$.
To see how this could arise in an imitation process, consider revision protocols of the form:
\begin{equation}
\label{eq:gen2step}
\rho_{ij} (F, x)= p_{ij} (F, x) r_{ij} (F, x), \text{ with } p_{ij}(F, x) = \lambda_{ij}(F, x) x_j
\end{equation}
for some positive functions $\lambda_{ij}$.
This models a two-step process: in step 1, a revising $i$-strategist gets interested
in strategy $j$ with a probability $p_{ij}$ that we call a \emph{selection rate}.
We allow it to depend on both payoffs and frequencies, but in our main examples,
it depends only on strategy frequencies;  in step 2, he adopts strategy $j$ with a
probability proportional to a quantity $r_{ij}$ that depends on payoff considerations,
and that we call an \emph{adoption rate}.\footnote{We allow these adoption rates to
depend on both payoffs and frequencies as we want to allow for protocols comparing
one's current payoff to, e.g., the average payoff in the population, which the vector
$F(x)$ alone does not allow to compute; nevertheless, we have in mind a payoff-based
second step.} The assumption $p_{ij}(F, x) = \lambda_{ij}(F, x) x_j$ just means that the
probability $p_{ij}$ to consider switching to strategy $j$ is zero whenever $x_j=0$, since we are modeling an imitation process.
Our adoption rates $r_{ij}$ will typically be monotonic, in the sense of Eq. \eqref{eq:monotonicity}.
Thus, the difference with monotone imitative dynamics is that the probability with which a revising agent gets interested in strategy $j$ need not be exactly $x_j$; that is, the $\lambda_{ij}$ need not be all constant and equal to $1$.
Here are some examples.

\begin{example}
\label{ex:list}
Meeting several agents and making a list of their strategies: a protocol advantaging rare strategies.
\end{example}
Assume that, in step 1, a revising agent does not meet one but $m$ other agents uniformly at random, where $m$ is a bounded random variable independent of the strategy played by the agent.
He then makes a list of the strategies they play, and selects at random a strategy in this list, as a candidate.
He might then learn more about this strategy's payoff,  by talking to the agent he met, by experimenting with this strategy for a short, un-modeled period of time, or by some thought experiment.
He then decides to adopt it or not according to a standard adoption rate $r_{ij}$.

As a concrete example, assume that the revising agent meets one agent playing strategy 1, two playing strategy 2 and two playing strategy 3.
He would then make a list of the strategies met: $\{1, 2, 3\}$, and pick each of them with the same probability, hence with probability $1/3$.\footnote{Picking up a strategy with a probability proportional to the number of agents met playing them (so here probabilities $1/5$, $2/5$, $2/5$) boils down to selecting a candidate uniformly at random, just breaking the selection process in two.
So this would lead to a neutral step 1.
For similar reasons, if $m=1$ or $m=2$, the above process leads to a neutral step 1.
This is why we need $m\geq 3$ with positive probability.} This is similar to protocols generating Smith or Brown-von Neumann-Nash dynamics, except that, instead of having a list of all possible strategies, an agent becomes aware of other possible strategies by meeting agents using them.

Provided that the number $m$ of agents met is equal to 3 or more with positive probability, the above step 1 advantages rare strategies compared to the reference case $p_{ij}(x)= x_j$, in the sense that the lower $x_j$, the higher the multiplicative factor $\lambda_{ij}$ in \eqref{eq:gen2step}. In other words, in proportion to their frequencies,
rare strategies are more often selected at step 1 than frequent strategies.
Another interpretation is as follows. Assume that after deciding which strategy to investigate, the revising agent obtains information about its payoffs by talking to a randomly selected mentor: one of the agents playing this strategy among those he met. Then if Alice plays a rarer strategy than Bob, she is (ex-ante) more likely to serve as a mentor.

\begin{proposition}\label{prop:ex1}
Assume $m \geq 3$ with positive probability.
Then in the first step of \cref{ex:list}, $p_{ij}(x) = x_j \lambda_j(x)$ where the functions $\lambda_j$
satisfy
\[\forall x \in X, \forall (j,k) \in I \times I, x_j < x_k  \Rightarrow \lambda_j(x) > \lambda_k(x)\]
\end{proposition}
\begin{proof}
See \cref{app:proofs}.
\end{proof}
We do not need step 1 to be exactly as described above.
Any protocol whose first step is a combination of the above one and a standard one ($p_{ij} =x_j$) would favour rare strategies in a similar sense.
Our results also apply to protocols that cannot be separated in two steps in the sense of Eq. \eqref{eq:gen2step}, but still favour rare strategies. This is discussed in \cref{app:moreprot}.
\begin{example}
\label{ex:maj}
Following the majority: a protocol advantaging frequent strategies.
\end{example}
As in the previous example, assume that a revising agent first meets $m$ other agents, where $m$ is a bounded random variable independent of the strategy played by the agent.
But now, he selects as a candidate the strategy played by the highest number of these agents, if there is only one.
If there are several such strategies, he selects one of these strategies uniformly at random.
Thus, if he meets one agent playing strategy 1, two playing strategy 2 and two playing strategy 3, then with probability 1/2 he selects strategy 2, and with probability 1/2, he selects strategy 3.

This step 1 advantages frequent strategies in the sense that the higher $x_j$, the higher the multiplicative factor $\lambda_{ij}$ (which here is independent of $i$).
In this sense, frequent strategies are imitated more often, or  more precisely,  more often selected at step 1.

\begin{proposition}\label{prop:ex2}
Assume that $m \geq 3$ with positive probability.
Then in the first step of \cref{ex:maj}, $p_{ij}(x) = x_j \lambda_j(x)$ where the functions $\lambda_j$  satisfy
\[\forall x \in X, \forall (j, k) \in I \times I, x_j < x_k  \Rightarrow \lambda_j(x) < \lambda_k(x)\]
\end{proposition}
\begin{proof}
See \cref{app:proofs}.
\end{proof}
As for \cref{ex:list}, a number of variants could be considered that cannot easily be put in the form \eqref{eq:gen2step}, but still favour frequent strategies, and to which our results would apply. %NEW
Note also that other forms of conformity biases have been studied in the literature on the evolution of cooperation, and shown to allow for the survival of cooperation in the prisoner's dilemma (Boyd and Richerson, 1988 \cite{3}; see also Eq. (1) in Heinrich and Boyd, 2001 \cite{4}, or in Pe\~na et al., 2009 \cite{8}).
\begin{example}
\label{ex:other}
Trying to meet agents playing other strategies than one's own: a protocol disadvantaging frequent strategies.
\end{example}
Assume that in step $1$, a revising agent of type $i$ meets somebody uniformly at random in the population.
If this person is of a type $j \neq i$, then the revising agent considers switching to $j$.
If this person is also of type $i$, then the revising agent tries again.
If after trying $m$ times, he did not manage to meet an agent of another type, he stops and keeps using strategy $i$.
The maximal number of trials $m$ could be a random variable.
We only assume that the law of this maximal number is the same for all strategies, that it is almost surely finite, and that with positive probability, it is equal to $2$ or more.

The motivation for such a behavior is that an agent currently playing strategy $i$ already knows that this is a possible behavior and already has a pretty good idea of how good this strategy is.
So talking with an agent of the same type is not very informative. Upon meeting an agent of the same type, a revising agent might thus be willing to try to meet somebody else.\footnote{If the payoff of a strategy is not deterministic, talking with other agents playing the same strategy is useful, but likely less so than talking to an agent with a different behaviour.}

For any $j \neq i$, the probability that a revising agent of type $i$ meets an agent playing another strategy for the first time at the $k^{th}$ trial, and that this agent is of type $j$, is $x_i^{k-1} x_j$.
So the probability $p_{ij}$ that  a revising agent of type $i$ considers switching to strategy $j$ is:
\[p_{ij} = x_j \lambda (x_i), \text{ with } \lambda(x_i) = 1 + x_i + \dotsm + x_i^{m-1}\]
The function $\lambda$ is strictly increasing. In this sense frequent strategies imitate more often than rare ones (or rather, are proportionally more likely to select another type at step 1). This is because agents from frequent types try on average more times to meet another type than agents from rare types.

This favours rare types but not in the same way as in \cref{ex:list}.
Indeed, the fact that a strategy is rare will not increase its chance to be considered for imitation, in the sense that if $j$ and $k$ are two strategies different from $i$, $p_{ij}/x_j = p_{ik}/x_k = \lambda(x_i)$, irrespective of the relative frequencies of strategies $j$ and $k$.
So $j$ and $k$ have the same ``extra-probability" of being selected by $i$.
In terms of the mother-equation \eqref{eq:mother}, the advantage of rare strategies is a higher inflow in \cref{ex:list} and a lower outflow in \cref{ex:other}.

The first step of \cref{ex:other} may also be interpreted as follows: the revising agent meets $m$ agents, keeps the same strategy if they all play as he does, and otherwise disregards all agents playing his strategy; he then picks up one of the remaining agents uniformly at random, and chooses her strategy as a candidate.
Thus, if he plays strategy 3 and meets one agent playing strategy 1, two playing strategy 2 and two playing strategy 3, he ends up choosing strategy 1 with probability 1/3 and strategy 2 with probability 2/3.

\begin{example}
\label{ex:confirmation}
Confirmation bias: a protocol favouring frequent strategies.
\end{example}

Assume that a revising agent meets $m$ other agents and that its main purpose is to be reassured that his strategy is not completely foolish.
More precisely, if at least one of the agents met plays the same strategy as he does, then he keeps it; otherwise, he selects uniformly at random one of the agents met and envisions to imitate her.
This leads to $$p_{ij} = (1-x_i)^m \frac{x_j}{1-x_i}= (1- x_i)^{m-1} x_j$$ for any $i \neq j$.
Thus, $\lambda_{ij}(x)=
(1- x_i)^{m-1}$. If $m \geq 2$, this expression is strictly decreasing in $x_i$, hence this protocol favours frequent strategies.
This is an example of frequent strategies imitating less often than rare strategies (or rather, being proportionally less likely to select another strategy as a candidate at step 1).

\subsection{A definition of favouring rare or frequent strategies}

Consider a two-step revision protocol of the form \eqref{eq:gen2step}:
\footnote{
Our results go through if all definitions in this section are restricted to the case where $i$ and $j$ are twin strategies, in that they have the same payoff function: $F_i = F_j$. This is because the strategy of the proof is to first use the advantage to rare or frequent strategies in a game with twin strategies, and then penalize one of them to make it dominated.}

\begin{definition} The first step is \emph{fair} is $\lambda_{ij}= 1$ for all $i \neq j$.
\end{definition}

\begin{definition}[being selected more often]
Per capita, rare strategies are more often selected at step 1 than frequent ones if for all $(F, x)$ and all strategies $i, j$ such that $x_i < x_j$, $\lambda_{ji}(F, x) \geq \lambda_{ij}(F,x)$, and $\lambda_{ki}(F,x) \geq \lambda_{kj}(F, x)$ for all strategies $k \notin\{i,j\}$.
They are selected strictly more often if these conditions hold with strict inequalities.
Frequent strategies are selected more often (in a weak or strict sense) if the same conditions hold when $x_i > x_j$.
\end{definition}

\begin{definition}[selecting other strategies less often]
Per capita, rare strategies select other strategies less often if for all $(F, x)$ and all strategies $i, j$ such that $x_i < x_j$, $\lambda_{ij} \leq \lambda_{ji}$ and for all strategies $k \notin\{i,j\}$, $\lambda_{ik} \leq \lambda_{jk}$.
They select other strategies strictly less often if these conditions hold with strict inequalities.
Frequent strategies select other strategies less often  (in a weak or strict sense) if the same conditions hold when  $x_i > x_j$.
\end{definition}

\begin{definition}[favouring rare or frequent strategies]
\label{def:adv}
Step 1 favours rare strategies if rare strategies are more often selected and select other strategies less often than frequent ones, and at least one of these properties holds strictly.
It favours frequent strategies if frequent strategies are more often selected and select other strategies less often, and at least one of these properties holds strictly.
\end{definition}
With this vocabulary, the protocols of Examples 1 and 3 both favour rare strategies, but not for the same reason.
In \cref{ex:list}, rare strategies are selected strictly more often
than frequent ones, while in \cref{ex:other}, they select other strategies strictly less often.
The protocols of Examples 2 and 4 favour frequent strategies.

%----------------------------------------------------------------------
%%% EXAMPLE
%----------------------------------------------------------------------
\section{A very simple example of survival of dominated strategies}
\label{sec:simple}

In this section, we consider two-step revision protocols \eqref{eq:gen2step} where in the second step, the adoption rates $r_{ij}$ are always positive.
This is the case in the imitation of success, in imitation driven by dissatisfaction, and in any generalization of the form $r_{ij} =  f(F_i) g(F_j)$ with $f$ and $g$ positive.\footnote{It would be natural to assume $f$ decreasing, $g$ increasing, %as then this means that in step 2, an agent playing strategy $i$ tends to adopt strategy $j$ with a probability which is decreasing in the payoff of strategy $i$ and increasing in the payoff of strategy $j$;
but this is not needed.}
For such protocols, as soon as the first step is not fair, survival of dominated strategies occurs in the simplest of games.
\begin{proposition}
\label{prop:simple} Consider dynamics generated by protocols such that the functions $\lambda_{ij}$ and $r_{ij}$ are jointly continuous in $(F, x)$, the adoption rates $r_{ij}$ are strictly positive, and $r_{ij}(F, x) = r_{ji}(F, x)$ whenever $F_i(x)= F_j(x)$.  Consider the $2 \times 2$ game $\Gamma^{\ep}$ with payoff function $F^{\ep} = (F_1^{\ep}, F_2^{\ep})$ given by $F^{\ep}_1(x)= 1$ and $F^{\ep}_2(x)= 1- \ep$, for all $x$ in $X$.
\begin{enumerate}
\item If the first step favours rare strategies, then for any $\alpha > 0$,
there exists $\bar{\ep}>0$ such that, for any $\ep \in [0 , \bar{\ep}]$
and for any initial condition $x(0)$ in $\mathrm{int}(X)$, $\liminf x_2(t) \geq 1/2 - \alpha$ as $t \to +\infty$.

\item If the first step favours frequent strategies, then for any $\alpha > 0$,
there exists $\bar{\ep}>0$ such that, for any $\ep \in [0 , \bar{\ep}]$
and for any initial condition $x(0)$ such that $x_2(0) \geq 1/2 + \alpha$, $x_2(t) \to 1$ as $t \to +\infty$.

\item If there exists $\hat{x} \in \mathrm{int}(X)$ such that $\lambda_{12}(F^0, \hat{x}) > \lambda_{21}(F^0, \hat{x})$, then
there exists $\bar{\ep}>0$ such that, for any $\ep \in [0 , \bar{\ep}]$,
for any initial condition such that $x_2(0) >\hat{x}_2$, $\liminf x_2(t) \geq \hat{x}_2$.
\end{enumerate}	
\end{proposition}

\begin{proof}
1) With only two strategies, the mother-equation \eqref{eq:mother} boils down to \[\dot x_1 = x_1(1-x_1) h(F, x) \text{ with } h(F, x)= \lambda_{21} r_{21} - \lambda_{12} r_{12}.\]  Our assumptions ensure that $h$ is jointly continuous.
In game $\Gamma^{0}$, $r_{21}= r_{12}$ for all $x$, hence $h(F^0, x)= (\lambda_{21} - \lambda_{12}) r_{12}$.
Since we assume $r_{12}>0$, $h(F^0, x)$ has the sign of $\lambda_{21} - \lambda_{12}$.
Thus, if step 1 favours rare strategies, $h(F^0, x) > 0$ if $0 \leq x_1 < 1/2$ and $h(F^0, x)< 0$ if $1/2 < x_1 \leq 1$.
Thus, in game $\Gamma^0$, $x_1(t) \to 1/2$ as $t \to +\infty$ for any interior initial condition.
Now let $\alpha \in (0, 1/2)$.
Since the sets $[0, 1/2 - \alpha]$ and $[1/2+ \alpha, 1]$ are compact, and $h$ is jointly continuous, it follows that for any $\ep >0$ small enough, in game $\Gamma^{\ep}$, we still have $h(F^{\ep}, x) > 0$ on $[0, 1/2 - \alpha]$ and $h(F^{\ep}, x) < 0$ on $[1/2 + \alpha, 1]$.  Therefore, in $\Gamma^{\ep}$, for any interior initial condition,
\[\frac{1}{2} - \alpha \leq \liminf_{t \to + \infty} x(t) \leq \limsup_{t \to + \infty} x(t) \leq \frac{1}{2} + \alpha.\]

2) Similar arguments show that, if step 1 favours frequent strategies, then $x_2(t) \to 1$ for any initial condition such that $x_2(0) >1/2$ in game $\Gamma^0$, and for any initial condition such that $x_2(0) \geq 1/2+ \alpha$ in $\Gamma^{\ep}$, provided that $\ep$ is small enough.

3) The assumption essentially amounts to assuming that the first step is not fair. We may then assume that there exists $\hat{x}$ such that $\lambda_{12}(F^0, \hat{x}) > \lambda_{21}(F^0, \hat{x})$.
Then in $\Gamma^{0}$, $h(F^0, \hat{x}) <0$, hence for any $\ep>0$ small enough, $h(F^{\ep}, \hat{x}) < 0$.
It follows that at $\hat{x}$, $\dot x_2 >0$.
Since the state space is a segment, the result follows.
 \end{proof}

\emph{How dominated can surviving strategies be?} As results of Hofbauer and Sandholm, the proof of \cref{prop:simple}
relies on arbitrarily small domination levels. It does not say whether strategies that are substantially dominated can survive. To tackle
this question, consider a game with only two strategies, $1$ and $2$, with constant payoffs: $F_1(x) = u_1$ and $F_2(x) = u_2 < u_1$ for all $x$ in $X$. For
a protocol of type \eqref{eq:gen2step}, there are at least as many transitions from strategy 2 to strategy 1 than from 1 to 2 (hence the
frequency of strategy 2 does not decrease) if and only if $\lambda_{12} r_{12} \geq \lambda_{21} r_{21}$, or equivalently
\begin{equation*}
%\label{eq:payfre}
\frac{r_{21}}{r_{12}} \leq \frac{\lambda_{12}}{\lambda_{21}}
\end{equation*}
The LHS may be seen as the ``payoff effect" and the RHS as the ``frequency effect".
This inequality takes a simple form if we assume
\begin{itemize}
\item $r_{ij} = u_j \geq 0$, as in the imitation of success.
\item $p_{ij}= x_j \lambda(x_i)$ with $\lambda(x_i) = 1 + x_i + \dotsm + x_i^{m-1}$, as in \cref{ex:other} from \cref{sec:ImProc}, where a revising agent tries to meet an agent playing another strategy up to $m$ times before giving up.
\end{itemize}
It is then easy to see that  the strictly dominated strategy $2$ survives whenever $u_2 >  u_1/m$. Moreover, in that case, $x_2(t) \to x_2^{\ast}$ where $x_2^{\ast}$ is the solution of
\[ u_2/u_1 = \frac{x_2(1- x_2^m)}{x_1 (1-x_1^m)} \text{ with } x_1 = 1-x_2.\]
 \cref{FigA} draws the value of the asymptotic frequency $x_2^{\ast}$ of the dominated strategy as a function of the ratio $u_2/u_1$, for various values of $m$. For instance, if $m=2$, the dominated strategy survives if its payoff is at least half the payoff of the dominant strategy $(u_2/u_1 \geq 1/2)$, its asymptotic frequency is larger than 0.2 if $u_2/u_1 \geq 2/3$, and larger than $1/3$ if $u_2/u_1 \geq 0.8$. Larger values of $m$ lead to even larger frequencies of the dominated strategy. Thus, at least for this protocol, relatively large differences in payoffs still allow for survival of strictly dominated strategies at significant frequencies.%\\

\begin{figure}
  \centering
    \includegraphics[width=.75\textwidth]{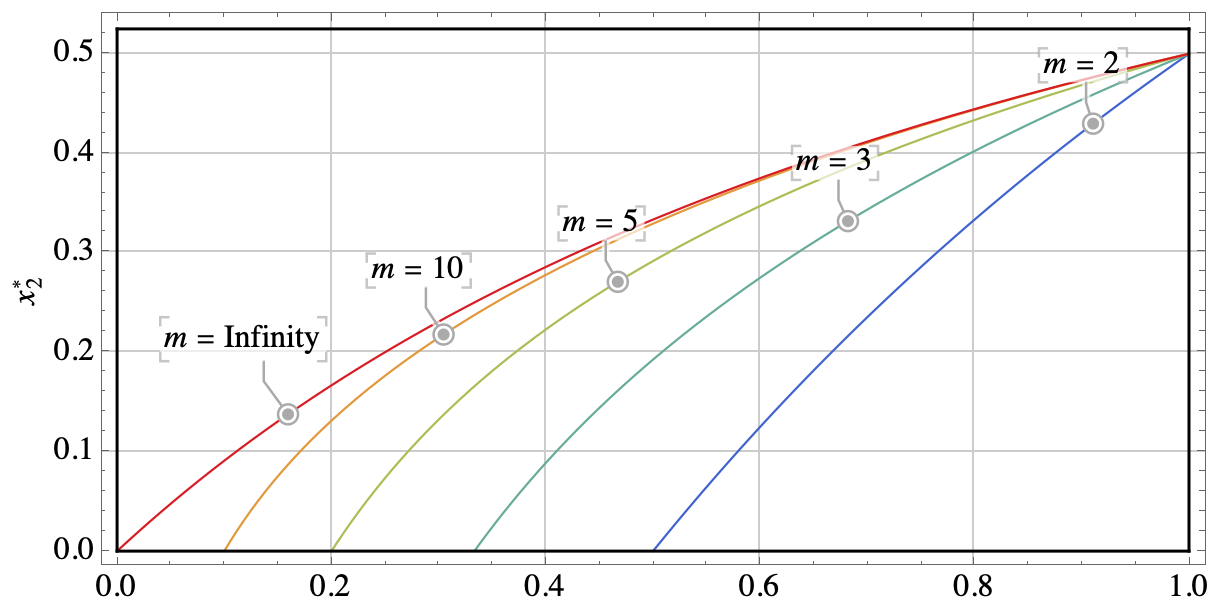}
  \caption{\textbf{Asymptotic frequency of  the dominated strategy as a function of the payoff ratio $u_2/u_1$ for various values of $m$.}}
  \label{FigA}
\end{figure}

%----------------------------------------------------------------------
%%% COMPARISON
%----------------------------------------------------------------------
\section{Imitation through comparison of payoffs}
\label{sec:paycomp}

In imitation protocols considered in the previous section, adoption rates are always positive, and rest-points correspond to an equilibrium between inflow and outflow, rather than an absence of strategy changes.  Though these adoption rates are standard, they have the debatable property that revising agents do not compare the payoff of their current strategy to the payoff of the strategy they envision to adopt (or the average payoff in the population). As a result, agents may switch to a strategy with currently lower payoffs than their own (or lower than average).

In this section, we show that survival of dominated strategies also occurs for adoption rates based on payoff comparison, such as $r_{ij} = [F_i - F_j]_+$, $r_{ij}= [F_j - \bar{F}]_+$, or generalizations thereof.\footnote{The examples we give cannot be of simple $2 \times 2$ games, as in the previous section. Indeed, in a game with only two strategies, such adoption rates prevent agents playing the dominant strategy to adopt the dominated one, so the dominated strategy gets extinct. This is also the case for any dynamics satisfying Positive Correlation (defined below).} To do so, we first need to show that, under mild additional assumptions, these protocols lead to dynamics satisfying the version of Positive Correlation for imitation processes:
\begin{equation}
\label{eq:PC}
\tag{PC$'$}
\sum_{i} \dot{x}_i F_i > 0
\end{equation}
whenever $x$ is not a population equilibrium, that is, a population state at which all strategies with a positive frequency get the same payoff (or in other words, a rest point of the replicator dynamics). An interpretation of  \eqref{eq:PC} is that, in a fixed environment, the average payoff in the population would increase, unless it is already maximal.\footnote{On top of replacing Nash equilibrium with population equilibrium, condition  \eqref{eq:PC} somehow combines the Positive Correlation condition of Hofbauer and Sandholm ($\sum_{i} \dot{x}_i F_i > 0$ whenever $\dot{x} \neq 0$) and their Nash Stationarity condition ($\dot{x} \neq 0$ whenever $x$ is not a Nash equilibrium).}

\subsection{Protocols leading to Positive Correlation}
Define the sign function by, for any real number $a$: $\mathrm{sgn}(a) = 1$ if $a>0$, $\mathrm{sgn}(a) = -1$ if $a <0$, and $\mathrm{sgn}(0)=0$.
\begin{proposition}
\label{prop:PC}
Consider dynamics arising from protocols of type \eqref{eq:gen2step}. Condition \eqref{eq:PC} is satisfied if at least one of the following properties holds:\footnote{The equalities below are between functions: $F_i$, $F_j$ may depend on $x$, and $r_{ij}$, $r_i$, $r_j$, $p_{ij}$, $\lambda_i$, $\lambda_j$ may depend on $(F, x)$.}
\begin{description}
\item[a)] (pairwise comparison) $\mathrm{sgn}(r_{ij}) = \mathrm{sgn}([F_j - F_i]_+)$.
\item[b)] (imitation of greater than average success)\footnote{If $f$ is constant, the second step is purely imitation of greater than average success. If $f$ is decreasing, it combines imitation of greater than average success with imitation driven by dissatisfaction.}\\
$p_{ij} = \lambda_j x_j$ with $\lambda_j$ positive; $r_{ij}= f(F_i) r_j$ with $f$ positive, nonincreasing, and $\mathrm{sgn}(r_j) = \mathrm{sgn}([F_j - \bar{F}]_+)$.

\item[c)] (imitation driven by less than average success)\footnote{If $g$ is constant, the second step is purely imitation driven by less than average success. If $g$ is decreasing, it combines imitation driven by less than average success with imitation of success.}\\
$p_{ij} = \lambda_i x_j $ with $\lambda_i$ positive;  $r_{ij}= g(F_j) r_i$ with $g$ positive, nondecreasing, and $\mathrm{sgn}(r_i) = \mathrm{sgn}([\bar{F}- F_i]_+)$.
\end{description}
 \end{proposition}
The intuition for this result is as follows: in case a), agents always switch to strategies with better payoff than their own; in case b), agents only switch to strategies $j$ earning more than $\bar{F}$, and for any such $j$, the average former payoff of agents switching to $j$ is no more than $\bar{F}$; in case c), agents only quit strategies $i$ earning less than $\bar{F}$, and for any such strategy $i$, on average, the new strategy of agents quitting $i$ earns at least $\bar{F}$.
It follows that in all three cases, in a fixed environment, the average population payoff would increase, which is one of the interpretations of condition \eqref{eq:PC}.
A formal proof of \cref{prop:PC} is given below.
\begin{proof}
We let the reader check that \[\sum_i \dot{x}_i F_i = \sum_{i, j} x_i \rho_{ij} (F_j - F_i)\] (intuitively, both sides represent the rate at which the average population payoff evolves in a fixed environment).

\vspace*{4pt}\noindent\textbf{Case a).} $\sum_i \dot{x}_i F_i =  \sum_{i, j} x_i p_{ij} r_{ij} (F_j - F_i)$ with $\mathrm{sgn}(r_{ij}) = \mathrm{sgn}([F_j - F_i]_+)$, so that $\mathrm{sgn}(r_{ij} (F_j - F_i)) = \mathrm{sgn}([F_j - F_i]_+)$.
It follows that the sum is zero if $F_i=F_j$ for any strategies $i$, $j$ such that $x_i>0$, $x_j>0$ (that is, at a population equilibrium) and positive otherwise.

\vspace*{4pt}\noindent\textbf{Case b).} Let $p_j = \lambda_j x_j$, with $\lambda_j >0$; let $\bar f= \sum_k x_k f(F_k)$ and let $y_i = x_i f(F_i) / \bar{f}$.
Note that $\sum_i y_i=1$.
We have:
%\begin{eqnarray*}
\begin{equation*}
\begin{split}
\sum_i \dot{x}_i F_i = \sum_{i, j} x_i f(F_i) \lambda_j x_j r_j (F_j - F_i) & =  \bar{f} \sum_{i, j} y_i \lambda_j x_j  r_j (F_j - F_i)\\
											& =  \bar{f} \sum_j \lambda_j x_j r_j (F_j- \sum y_i F_i).
\end{split}
\end{equation*}
%\end{eqnarray*}
Since $f$ is nonincreasing, the $y_i$ (which may be thought of as distorted frequencies) give more weight to strategies with low payoffs than the true frequencies $x_i$, and it may be shown that $\sum y_i F_i \leq \sum x_i F_i = \bar{F}$.
Since $\mathrm{sgn}(r_j) = \mathrm{sgn}([F_j - \bar{F}]_+)$, it follows that we also have $sgn (r_j (F_j- \sum y_i F_i))= \mathrm{sgn}([F_j - \bar{F}]_+)$.
Thus, the whole sum is zero at population equilibria and positive otherwise.

\vspace*{4pt}\noindent\textbf{Case c).} Similarly, let $\bar{g}  = \sum_k x_k g(F_k)$ and $y_i = x_i g(F_i) / \bar{g}$.
We get:
%\begin{eqnarray*}
\begin{equation*}
\begin{split}
\sum_i \dot{x}_i F_i = \sum_{i, j} \lambda_i x_j r_i  g(F_j) (F_j - F_i) 	& = \bar{g} \sum_{i, j} \lambda_i r_i y_j (F_j - F_i)\\
												&  = \bar{g} \sum_{i} \lambda_i r_i \left(\left[\sum_j y_j F_j\right] - F_i\right).
\end{split}
\end{equation*}												
%\end{eqnarray*}
Since $g$ is nondecreasing, $\sum_j y_j F_j \geq \bar{F}$.
Moreover, $r_i$ has the sign of $[\bar{F} - F_i]_+$.
Therefore,  $r_i ([\sum_j y_j F_j] - F_i)$ has the sign of $[\bar{F} - F_i]_+$.
It follows that the whole sum is zero at population equilibria and positive otherwise.
\end{proof}

\subsection{Survival result}
Our results on survival of dominated strategies also hold for revision protocols that are not of the two-step form \eqref{eq:gen2step}.
To emphasize this fact, we first state a theorem with assumptions directly on the vector field $V^F$ and the switching-rates $\rho_{ij}$.
We then provide sufficient conditions for these assumptions to be satisfied by two-step revision protocols of form \eqref{eq:gen2step}.
We begin with a list of definitions and assumptions.

\begin{definition}
Strategies $i$ and $j$ are twins if for all $x$ in $X$, $F_i(x)=F_j(x)$.
\end{definition}
\begin{definition}
At a given population state of a given game:  strategy $i$ imitates other strategies if there exists $j \neq i$ such that $\rho_{ij} >0$; it is imitated by other strategies if there exists $j \neq i$ such that $\rho_{ji} >0$.
\end{definition}

On top of condition \eqref{eq:PC}, we will need the following assumptions:

\emph{Continuity (C)}: the vector field $V^F$ is Lipschitz continuous in $x$ and continuous in $u$ (implying joint continuity);
the functions $x \to \rho_{ij}(F, x)$ are continuous in $x$.

\emph{Imitation (Im)}: at any interior population state that is not a Nash equilibrium, each strategy $i$ imitates other strategies or is imitated by  other strategies (or both).

We also need either Advantage to Rarity or Advantage to Frequency, as defined below:

\emph{Advantage to Rarity (AR)}: in the interior of the simplex, if strategy $i$ and $j$ are twins, then
$\frac{\dot{x}_i}{x_i} \geq  \frac{\dot{x}_j}{x_j}$ whenever $x_i < x_j$. Moreover, at least one of the following additional properties holds:  \\
(AR1) The inequality is strict whenever at least one of the strategies $i$ and $j$ imitates other strategies.\\
(AR2) The inequality is strict whenever at least one of the strategies $i$ and $j$ is imitated by other strategies.

\emph{Advantage to Frequency (AF)}: idem but when $x_i > x_j$ instead of $x_i < x_j$.

\begin{theorem}
\label{th:hypno}
Fix $\eta >0$.
Assume that conditions \eqref{eq:PC}, (Im), and (C) are satisfied. If (AR) is satisfied (respectively, (AF)), then there exist 4-strategy games in which pure strategy $3$ strictly dominates pure strategy $4$ but $\liminf x_4(t) > \frac{1}{2} - \eta$ (respectively, $1- \eta$) for a large, open set of initial conditions.%\footnote{By a ``large set", we mean the whole simplex (respectively, the half-simplex defined by $x_4 \geq x_3$), except an arbitrarily small neighborhood of its boundary and of a line segment in the case of an advantage to rarity (respectively, to frequency).}
\footnote{By a ``large set", we mean the whole simplex (for an advantage to rarity) or the half-simplex defined by $x_4 \geq x_3$ (for an advantage to frequency), except an arbitrarily small neighborhood of its boundary and of a line segment.}

\end{theorem}
The proof is given in the next section. It is based on ideas of Hofbauer and Sandholm. We first provide sufficient conditions for the assumptions of \cref{th:hypno} to hold. Consider a two-step revision protocol $\rho_{ij}(F, x) = x_j \lambda_{ij}(F, x) r_{ij}(F, x)$.
\begin{definition}
Step 2 treats twins identically if for any twin strategies $i$ and $j$, $r_{ij}= r_{ji}$ and for any $k \notin \{i, j\}$, $r_{ik} = r_{jk}$ and $r_{ki}= r_{kj}$.
\end{definition}
\begin{proposition} Consider dynamics generated by a two-step protocol of form \eqref{eq:gen2step} satisfying the assumptions of \cref{prop:PC}.
Then \cref{th:hypno} applies provided that both of the following conditions hold:\\
a) the functions $\lambda_{ij}$ and $r_{ij}$ are continuous, and Lipschitz continuous in $x$;\\
b) the selection rates $\lambda_{ij}$ are strictly positive,  step 1 favours rare (respectively frequent) strategies, and step 2 treats twins identically.
\end{proposition}

\begin{proof}
The conditions of \cref{prop:PC} imply \eqref{eq:PC} and (Im), as would any protocol based on adoption rates $r_{ij}$ with the same sign as $[F_j - F_i]_+$, or $[F_j - \bar{F}]_+$. Assumption a) implies (C). It remains to show that b) implies (AR) (or, respectively, (AF)). Let $i$ and $k$ be twin strategies.
We let the reader check that, since step 2 treats twins identically:
\begin{equation*}
%\label{eq:big}
\frac{\dot{x}_i }{x_i} - \frac{\dot{x}_j }{x_j}= \sum_{k \notin \{i, j\}} x_k r_{ki} (\lambda_{ki} - \lambda_{kj})  + \sum_{k \notin \{i, j\}} x_k r_{ik} (\lambda_{jk}- \lambda_{ik})  +  r_{ij} (x_j + x_i) [\lambda_{ji} -  \lambda_{ij} ]
\end{equation*}
Moreover, again because step 2 treats twins identically, the assumption in (AR1) that at least one of the strategies $i$ and $j$ imitates (or, in (AR2), is imitated by) other strategies boils down to the fact that this holds for strategy $i$.
Now assume that $x_i < x_j$ and that step 1 favours rare strategies. Then all three terms in the RHS are nonnegative. There are two cases.

\vspace*{4pt}\noindent\textbf{Case 1.} If rare strategies are more often selected at step 1. Then $\lambda_{ki} > \lambda_{kj}$ for all $k \notin \{ i, j\}$, and $\lambda_{ji} >  \lambda_{ij}$.
Provided that strategy $i$ is imitated by other strategies, it follows that the first or the third term, hence the whole RHS, is positive. Therefore (AR2) holds, hence (AR) holds.

\vspace*{4pt}\noindent\textbf{Case 2.} Otherwise, rare strategies select other strategies less often. The second or third term in the RHS are then positive, provided that strategy $i$ imitates other strategies. Therefore (AR1) holds, hence (AR) holds as well.

Similarly, if step 1 favours frequent strategies, condition (AF)  is satisfied. This concludes the proof.
\end{proof}

%----------------------------------------------------------------------
%%% PROOF
%----------------------------------------------------------------------
\section{Proof of \cref{th:hypno}}
\label{sec:proof}
The proof combines ideas of the proofs of Hofbauer and Sandholm's (2011) Theorems 1 and 2. As in their Theorem 2, the game used is the hypnodisk game with a feeble twin. As in their Theorem 1, in the case of an advantage to rarity, the shares of strategies that always earn the same payoff tend to become equal.

\subsection{The game} We first briefly recall the construction of the hypnodisk game with a feeble twin (see also Figures 5, 6, 7 in Hofbauer and Sandholm). The construction has three steps. Below, $X$ may denote the simplex of a game with three or four strategies, depending on the context.

\vspace*{4pt}\noindent\textbf{Step 1. The hypnodisk game.} The hypnodisk game is a 3-strategy game, with nonlinear payoffs: it is not the mixed extension of a finite game. It may be seen as a generalization of Rock-Paper-Scissors, in that it generates cyclic dynamics for any dynamics satisfying Positive Correlation. Its payoff function will be denoted by $H$. We refer to Hofbauer and Sandholm for a precise definition and analysis of this game. The important properties are the following:

a) there is a unique Nash equilibrium $p = (1/3, 1/3, 1/3)$.

b) there exist two reals numbers $r$ and $R$ with $0 < r < R < 1/\sqrt{6}$ such that: within the disk of center $p$ and radius $r$, the payoffs are as in a coordination game: $H_i(x) = x_i$; outside of the disk of center $p$ and radius $R$, the payoffs are as in an anti-coordination game: $H_i(x) = -x_i$. These disks will be denoted by $D_r = \{x \in X, || x - p||_2 < r\}$ and $D_R = \{x \in X, || x - p||_2 \leq R\}$.\footnote{We define $D_r$ as an open disk so that the annular region $D_R \backslash D_r$ is closed.}

c) In the annular region with radii $r$ and $R$, the payoffs are defined in a way that preserves the regularity of the payoff function.

d) The radii $r$ and $R$ may be chosen arbitrarily small if useful.

The payoff function $F$ is a map from $X \subset \mathbb{R}^3$ to $\mathbb{R}^3$ and may be seen as a vector field. Property b) implies that the projection of this payoff vector field on the affine span of the simplex points towards the equilibrium outside of the larger disk $D_R$, and away from the equilibrium within the smaller disk $D_r$ (except precisely at the equilibrium).\footnote{The idea to preserve the regularity of the payoff function, i.e.,  property c), is to rotate continuously (the projection of) the payoff vector field so that it rotates by 180 degrees in total in the annular region, see Hofbauer and Sandholm.} Moreover, the geometric interpretation of condition \eqref{eq:PC} is that, except at population equilibria, the payoff vector field, or equivalently, its projection on the affine span of the simplex, makes an acute angle with the dynamics' vector field $V^F$. It follows that in the hypnodisk game, for any dynamics satisfying \eqref{eq:PC} and any interior initial condition different from the Nash equilibrium, the solution eventually enters the annulus region with radii $r$ and $R$ and never leaves (Hofbauer and Sandholm, Lemma 3).

A similar construction could be made but putting the unique equilibrium at any desired place in the interior of the simplex instead of the barycenter.\footnote{\label{ft20} The disks $D_r$ and $D_R$ would then surround the equilibrium and the projected payoff vector field would point towards the equilibrium outside of the larger disk $D_R$, and away from it inside of the smaller disk $D_r$.  This is the case for instance if $H_i(x) = p_i - x_i$ outside $D_R$ and $H_i(x)=x_i - p_i$ inside $D_r$, where $p$ is the equilibrium.}

\vspace*{4pt}\noindent\textbf{Step 2. Adding a twin.} Let us now add a fourth strategy that is a twin of the third. This leads to a 4-strategy game, which is called the hypnodisk game with a twin. Its payoff function $F$ satisfies: for any $x$ in $X$, $F_i(x) = H_i(x_1, x_2, x_3 + x_4)$ for $i=1,2,3$ and $F_4(x)= F_3(x)$. There is now a segment of Nash equilibria:
$$\mathrm{NE}=\{x \in X, (x_1, x_2, x_3 + x_4) = (1/3, 1/3, 1/3)\}.$$
The disks $D_r$ and $D_R$ become intersections of cylinders and of the simplex, which are denoted by $I$ and $O$ (for Inner and Outer cylinders):
$$I = \{x \in X, (x_1, x_2, x_3 + x_4) \in D_r\}; \quad O = \{x \in X, (x_1, x_2, x_3 + x_4) \in D_R\}.$$
The annular area with radii $r$ and $R$ becomes the intercylinder region
\[D= O\backslash I =  \{x \in X, r^2 \leq (x_1- 1/3)^2 + (x_2- 1/3)^2 + (x_3 + x_4 - 1/3)^2 \leq R^2\}.\]
For any dynamics satisfying (C) and \eqref{eq:PC} and any interior initial condition not in $\mathrm{NE}$, the solution eventually enters this intercylinder zone, and then never leaves it (Hofbauer and Sandholm, Lemma 4):
$\exists T, \forall t \geq T, x(t) \in D$.

\vspace*{4pt}\noindent\textbf{Step 3. The feeble twin.} We now subtract $\ep>0$ from the payoffs of strategy $4$, so that it is now dominated by strategy $3$. This leads to the hypnodisk game with a feeble twin, which we denote by $\Gamma_{\ep}$.

\subsection{Sketch of proof of \cref{th:hypno}} Before providing a formal proof, we describe its logic. Consider first the hypnodisk game with an exact twin $\Gamma_0$. In the case of an advantage to rare strategies, the shares of strategy $3$ and $4$ tend to become equal. As a result, for any interior initial condition, solutions converge to an attractor $A$ which is contained in the intersection of the intercylinder region $D$ and the plane $x_3=x_4$. In this attractor, $\liminf x_4 \geq \frac{1}{6}- \frac{R}{\sqrt{6}}$. Because the vector field of the dynamics is jointly continuous in $(F, x)$, for $\ep> 0$ small enough, there is an attractor $A^{\ep}$ included in an arbitrarily small neighborhood of $A$, and whose basin of attraction is at least the old basin of attraction minus an arbitrarily small neighborhood of the union of the segment of $\mathrm{NE}$ and of the boundary of the simplex. It follows that for most initial conditions, $\liminf x_4 \geq  \frac{1}{6}- \frac{R}{\sqrt{6}} - \delta(\ep)$, with $\delta(\ep) \to 0$ as $\ep \to 0$.
Thus, if we fix any $\eta >0$, for $R$ and $\ep$ small enough, $\liminf x_4 \geq \frac{1}{6}- \eta$. We can get an ever larger value of $\liminf x_4$ with the same construction and proof, just replacing the standard hypnodisk game by a variant with unique equilibrium $(\beta, \beta, 1- 2\beta)$, see footnote \ref{ft20}.
We then get for $\beta$, $R$ and $\ep$ small enough, $\liminf x_4 \geq \frac{1}{2}- \eta$.\footnote{We thank Vianney Perchet for pointing this out to us.}

The case of an advantage to frequent strategies is similar, with some twists. Now in $\Gamma_0$, for any interior initial condition with $x_4 >  x_3$, the solution converges to an attractor $A'$ included in the intersection of the intercylinder region $D$ and of the plane $x_3=0$. In $\Gamma_{\ep}$, for $\ep$ small enough, there is an attractor included in an arbitrarily small neighborhood of $A'$, and whose basin of attraction is at least the basin of attraction of $A'$ minus a zone with an arbitrarily small Lebesgue measure. This allows to show that, for any $\eta>0$, we may find a game such that for many initial conditions (including all initial conditions such that $x_4 > x_3 + \eta$ and $x$ is not in the $\eta$-neighborhood of the union of the segment of Nash equilibrium and of the boundary of the simplex), for $\ep$ and $R$ small enough, $\liminf x_4 \geq 1/3 - \eta$.
By changing the equilibrium of the initial hypnodisk game, we get $\liminf x_4 \geq 1 -  \eta$.

\subsection{Formal proof of \cref{th:hypno}}
We now provide a formal proof. To fix ideas, let us assume that (AR) holds, and that the advantage to rarity is strict when at least one of the twin strategies
imitate other strategies (condition (AR1)). Other cases are similar. Consider game $\Gamma_0$ and fix an interior initial condition $x(0) \in \mathrm{NE}$.
As in Hofbauer and Sandholm, Lemma 4, we first obtain:
\begin{claim}
\label{cl:IC}
There exists a time $T$ such that for all $t \geq T$, $x(t)$ is in the intercylinder region $D$.
\end{claim}
\begin{proof} Since Hofbauer and Sandholm do not provide a formal proof, we do it here. Due to condition (PC'), the vector field $V^F(x)$ at the boundary of region $D$ points inwards, it follows that once solutions enter region $D$, they cannot leave it. By contradiction, assume that this is never the case, that is, the solution remains in the compact set $K = X\backslash \mathrm{int}(D)$, where $\mathrm{int}(D)$ denotes the relative interior of $D$. It follows that the solution has accumulation points in $K$, which cannot be on  $\mathrm{NE} \cup \mathrm{\mathrm{Bd}}(X)$. Moreover, the Euclidean distance $W(x)$ to the segment of Nash equilibria evolves monotonically (it increases within inner cylinder $I$ and decreases outside outer cylinder $O$). By a standard result on Lyapunov functions, all such accumulation points $x^{\ast}$ satisfy $\nabla W(x^*) \cdot F(x^*)=0$ (thus, if at time $t$, $x(t) = x^*$, then $dW(x(t))/dt=0$). But by construction, there are no such points in $K \backslash (\mathrm{NE} \cup \mathrm{Bd}(X))$, a contradiction.
\end{proof}
Moreover, as in Theorem 1 of Hofbauer and Sandholm:
\begin{claim}
\label{cl:equal}
$x_4(t)/x_3(t) \to 1$ as $t \to+\infty$.
\end{claim}
\begin{proof}

Let $V(x) = x_4/x_3$ and let $\dot V(x) = \nabla V(x) \cdot F(x)$ so that $\mathrm{\frac{d}{dt}} V(x(t)) = \dot{V}(x(t))$.
Due to condition (AR), $V(x(t))$ evolves (weakly) monotonically in the direction of $1$. Thus, assuming to fix ideas $x_4(0) < x_3(0)$, $V(x(t))$ is increasing and less than $1$, hence has a limit $l$ such that $V(x(0)) \leq l \leq 1$. Assume by contradiction that $l < 1$.
Let $K_i= \{x \in X \, | \, \rho_{ik}=0, \forall k \neq i\}$ be the set of population states at which strategy $i$ does not imitate any other strategy. Let
$$K = K_3 \cap K_4 \cap D \cap \{x \in X, x_4 = l x_3\}.$$
Note that $K$ is compact (by Continuity) and contained in the interior of the simplex (since in $D$, $x_1>0$, $x_2>0$, $x_3+ x_4>0$, and $l \neq 0$). We want to show that the solution cannot stay in $K$ forever. For any population state in $K$, strategies $3$ and $4$ do not imitate other strategies. Moreover, the state is not an equilibrium. So by Imitation, strategies $3$ and $4$ are imitated. Therefore, $\dot{x}_3 + \dot{x}_4 > 0$.
By Continuity and compactness of $K$, there exists $\ep >0$ and an open neighborhood $U$ of $K$ such that, whenever $x(t) \in U \cap X$, $\dot{x_3}+ \dot{x_4} > \ep$.
It follows that $x(t)$ cannot stay for ever in $U$, hence must have accumulation points in $X \backslash K$.

We now prove that this is impossible. Indeed, let $x^{\ast} \in X \backslash K$ be an accumulation point of $x(t)$.
Necessarily, $x^{\ast} \in D \cap \{x \in X \, | \, x_4 = l x_3\} \subset \mathrm{int}(X)$.
Moreover, by standard results on Lyapunov functions, $\dot{V}(x^{\ast})=0$. Since $x^{\ast} \in \mathrm{int}(X)$, it follows from (AR1) that $x^{\ast} \in K_3 \cap K_4$, so that $x^{\ast} \in K$. We thus get a contradiction.
This concludes the proof.
\end{proof}

Let $K_{\alpha}$ denote the compact set $X \backslash N_\alpha(\mathrm{NE} \cup \mathrm{Bd}(X))$, where $N_{\alpha}$ refers to the open $\alpha$-neighborhood for the Euclidean norm. Let $\ep \in (0, 1)$  and let $$U_{\ep} = \{x \in N_{\ep}(D), |x_4/x_3 - 1| < \ep\}.$$
Let $\Phi_t$ denote the time $t$ map of the flow; that is, $\Phi_t(x_0)$ is the value at time $t$ of the solution such that $x(0) = x_0$.
\begin{claim}
\label{cl:flow}
There exists $T$ such that for all $t \geq T$, $\Phi_t(K_{\alpha}) \subset U_{\ep}$.
\end{claim}

\begin{proof} Since the solution cannot leave $U_{\ep}$ it suffices to show that there exists $T$ such that $\Phi_T(K_{\alpha}) \subset U_{\ep}$.
Assume that this is not the case.
Then we may find a increasing sequence of times $t_n \to +\infty$ and a sequence of positions $x_n \in K_{\alpha}$ such that $\Phi_{t_n}(x_n) \notin U_{\ep}$.
By compactness of $K_{\alpha}$, up to considering a subsequence, we may assume that $x_n$ converges towards some $x_{\lim}$ in $K_{\alpha}$.
But by the previous claims, there exists a time $\tau$ such that $\Phi_{\tau} (x_{\lim}) \in U_{\ep/2}$.
 By continuity of the flow, there exists a neighborhood $\Omega$ of $x_{\lim}$ such that $\Phi_{\tau}(\Omega) \subset U_{\ep}$, hence  $\Phi_{t}(\Omega) \subset U_{\ep}$ for all $t \geq \tau$, since solutions cannot leave $U_{\ep}$ in forward time. But for $n$ large enough, $t_n \geq \tau$, $x_n \in \Omega$ but $\phi_{t_n}(x_n) \notin U_{\ep}$, a contradiction.\end{proof}

We now need to define $\omega$-limits, attractors and basins of attraction.

\begin{definition}[$\omega$-limit]
The \emph{$\omega$-limit} of a set  $U \subset X$ is defined as $\omega(U) = \bigcap_{t > 0} \mathrm{cl}(\phi^{ [t, \infty) } (U))$,
where for $T \subset \mathbb{R}$,
we let $\phi^T(U) = \cup_{t \in T} \phi^t (U)$. If $x \in X$, we write $\omega(x)$ instead of  $\omega(\{x\})$.
\end{definition}

\begin{definition}[attractor and basin of attraction]
A set $A \subset X$ is an \emph{attractor} if there is a neighborhood $U$ of $A$ such that $\omega(U) = A$. Its \emph{basin of attraction} is then defined as $B(A) = \{x : \omega(x) \subseteq A\}$.
\end{definition}

\begin{claim}
\label{cl:attractor}
Fix $\alpha>0$ small enough.
Then $A= \omega(K_{\alpha})$ is an attractor, it is included in the intersection of the intercylinder zone D and the plane $x_3 = x_4$, and its basin of attraction is $B(A)= \mathrm{int}(X)\backslash \mathrm{NE}$.
\end{claim}
\begin{proof} By \cref{cl:flow}, there exists a time $t >0$ such that $\phi_t(K_{\alpha}) \subset \mathrm{int}(K_{\alpha})$.
It follows (see Appendix A in Hofbauer and Sandholm) that $A$ is an attractor.
 By letting $\ep$ go to zero in \cref{cl:flow}, we obtain that
 $$A \subset \cap_{\ep >0} U_{\ep} = U_0 = D \cap \{x \in X : x_3 = x_4\}.$$
Finally, by \cref{cl:IC,cl:equal}, for all $x$ in $\mathrm{int}(X)\backslash \mathrm{NE}$, the solution starting in $x$ enters $K_{\alpha}$.
Therefore $\omega(x) \subset \omega(K_{\alpha}) =A$, hence  $(\mathrm{int}(X))\backslash \mathrm{NE} \subset B(A)$.
The reverse inclusion is obvious.
Note that $\omega(K_{\alpha})$ does not depend on $\alpha$ (as long as $\alpha$ is small enough).
\end{proof}

\begin{claim}  Call $\Gamma_{\ep}$ the hypnodisk game with an $\ep$-feeble twin. Let $\eta> 0$. For all $\ep>0$ small enough, in $\Gamma_{\ep}$, there is an attractor $A_{\ep} \subset N_{\eta}(A)$ whose basin of attraction includes $B(A) \backslash N_{\eta} (\mathrm{NE} \cup \mathrm{Bd}(X)) = X\backslash N_{\eta} (\mathrm{NE} \cup \mathrm{Bd}(X))$.
\end{claim}
\begin{proof} This follows from \cref{cl:attractor} and Continuity, as in Hofbauer and Sandholm (2011) \cite{5}.
\end{proof}

We now conclude: for $\ep$ small enough, from most initial conditions, solutions converge to an attractor along which $x_4$ is bounded away from zero. The minimum of $x_4$ along this attractor may be made higher than $1/6- R/\sqrt{6} - \eta$, where $R$ is the radius of the outer cylinder, which may be chosen arbitrarily small.
By taking as base game an hypnodisk game with an equilibrium such that $x_3$ is sufficiently close to $1$ (see footnote \ref{ft20}), we may transform $1/6$ in any number strictly smaller than $1/2$, and obtain $\liminf x_4 \geq 1/2 - \delta$ for any $\delta>0$ fixed beforehand.\footnote{For an advantage to frequent strategies, we get initially $\liminf x_4 \geq 1/3- R - \eta$ and then $\liminf x_4 \geq 1- \delta$.}

%----------------------------------------------------------------------
%%% DISCUSSION
%----------------------------------------------------------------------
\section{Discussion}
\label{sec:disc}
%\paragraph{The hypnodisk game}

\emph{The hypnodisk game. }
The hypnodisk game with a feeble twin is easy to analyze, and allows to prove survival results for large classes of dynamics. However, numerical simulations show that pure strategies strictly dominated by other pure strategies also survive in more standard games. \cref{FigC} illustrates imitation dynamics in a Rock-Paper-Scissors-Feeble Twin game for two different domination margins (the game is the same as in the numerical explorations of Hofbauer and Sandholm, Section 5.2):
\begin{equation}
\label{eq:RPST}
\begin{array}{c}
R \\ P \\ S \\ FT
\end{array}
\left(\begin{array}{cccc}
0 & -2 & 1 & 1 \\
1 & 0 & -2 & -2 \\
-2 & 1 & 0 & 0 \\
-2 - d & 1 - d & - d & -d
\end{array}\right)
\end{equation}
The dynamics are derived from a two-step protocol of form \eqref{eq:gen2step}, with a first step as in \cref{ex:other}
(trying to meet an agent playing another strategy),
with $m= 4$, and a second step based on payoff comparison: $r_{ij} = [F_j - F_i]_+$.
\begin{figure}
  \centering
    \includegraphics[width=0.48 \textwidth]{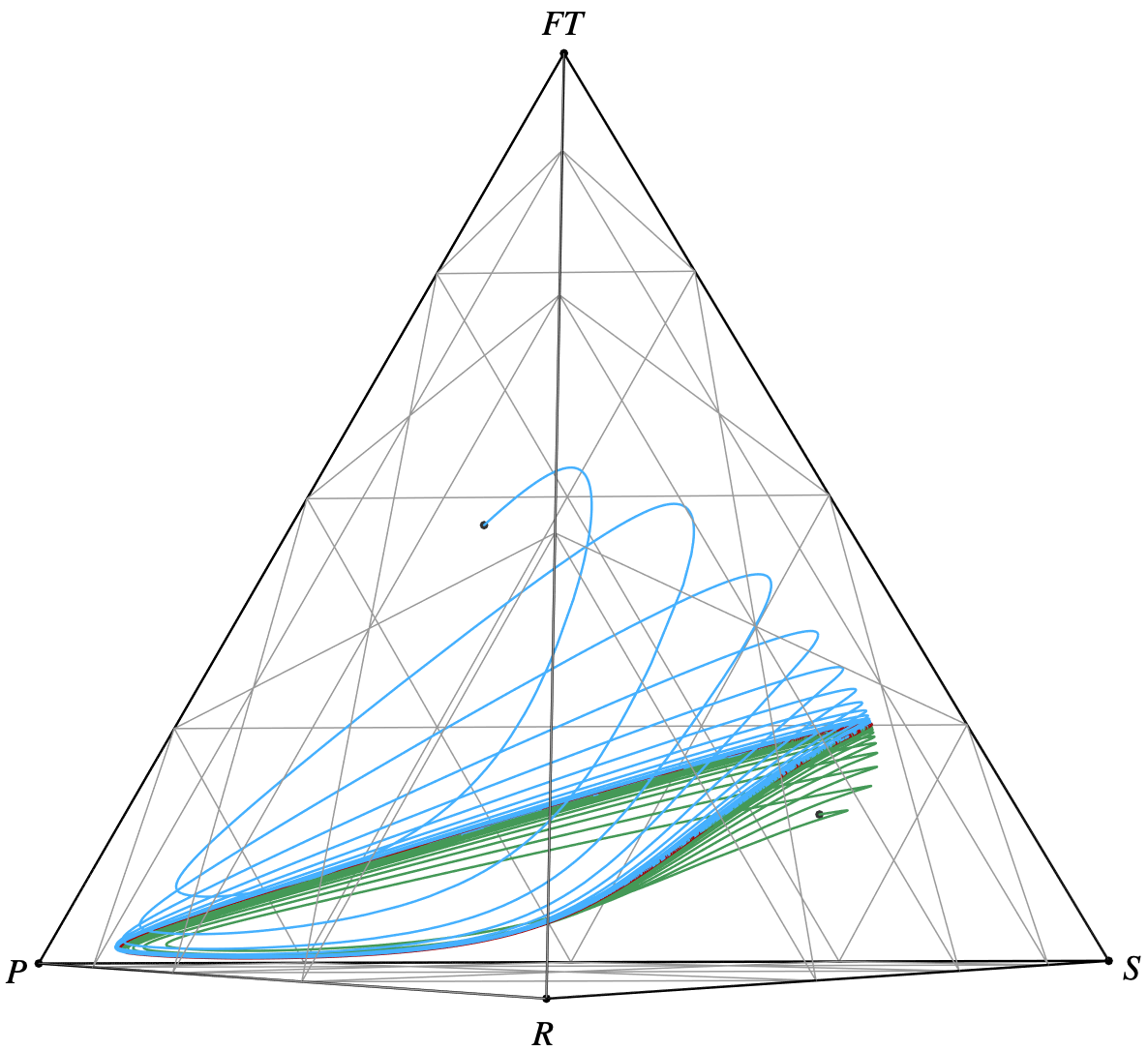}
    \hfill
    \includegraphics[width=0.48 \textwidth]{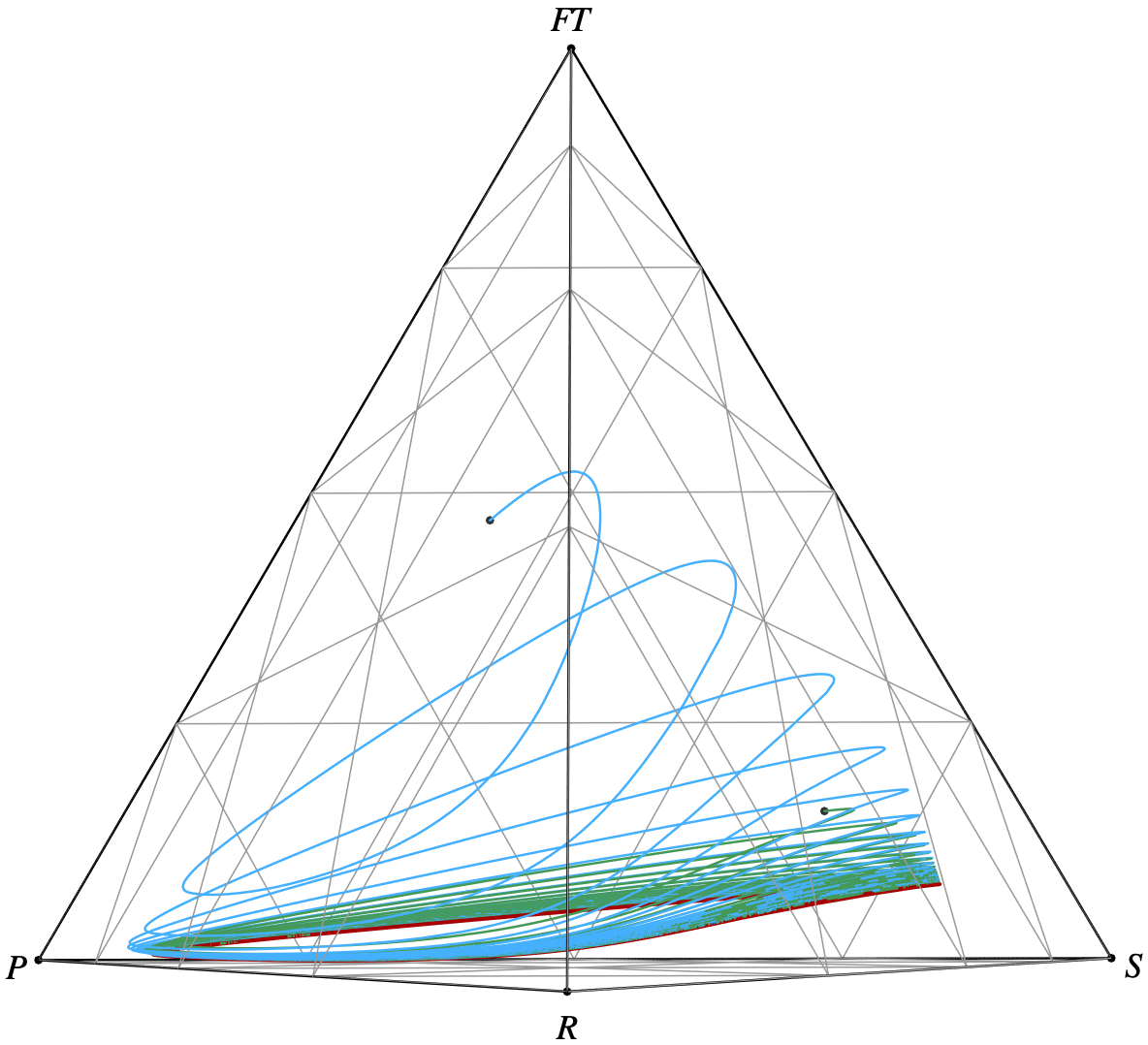}
  \caption{\textbf{Imitation dynamics in Game \eqref{eq:RPST}.} Left panel: $d = 0.04$; right panel: $d = 0.08$.
  In blue and green, two orbits starting respectively at $(1/7, 2/7, 1/7, 3/7)$ and $(1/7, 1/7, 4/7, 1/7)$. In red (hardly visible on the left), what appears to be a common limit cycle. The dynamics are described in the main text.}
  \label{FigC}
  \vspace{-2ex}
\end{figure}

%\paragraph{Monotone dynamics}
\emph{Monotone dynamics. }
Monotone dynamics (or imitative dynamics, in the sense of Sandholm) have long been known to eliminate pure strategies strictly dominated by other pure strategies. With our vocabulary, this may be formulated as follows: in a two-step protocol of form \eqref{eq:gen2step}, if Step 1 is fair ($p_{ij} = x_j$) and Step 2 is monotonic (in the sense of Eq.\eqref{eq:monotonicity}), then pure strategies strictly dominated by other pure strategies go extinct. Obviously, if step 1 is fair but step 2 is not monotonic, there is no reason to expect dominated strategies to go extinct. What we showed is that, similarly, when step 2 is monotonic, but step 1 is not fair, dominated strategies may survive.

%\paragraph{Elimination results are not robust}
\emph{Elimination results are not robust. }
For imitative dynamics, the elimination of strictly dominated pure strategies in all games relies on the fact that two strategies with the same payoff have the same per capita growth rate. This condition is an equality, and contrary to strict inequalities, equalities are not robust to small perturbations. In a sense, Hofbauer and Sandholm show that the elimination result is not robust to the introduction of the possibility to innovate. We show that it is not robust either to perturbations of the imitation protocol (here, perturbations of the first step), even if the dynamics still model pure imitation. See also Section 5.3. in Hofbauer and Sandholm.

%\paragraph{Inflow towards a dominated strategy}
\emph{Inflow towards a dominated strategy. }
At all times, some of the agents quit playing the dominated strategy for the dominating one, or some currently even better strategy. So for the dominated strategy to survive, it is needed that, to compensate, some other strategies keep imitating it.
This can occur in two ways:
\begin{enumerate}
\item
If solutions converge to a rest-point, but there is nonetheless a perpetual flow between strategies. That is, rest-points correspond to a macroscopic equilibrium between inflow and outflow, not an absence of strategy changes at the micro level (\cref{sec:simple}). This is not the case for protocols based on standard payoff comparison.
\item
If solutions do not converge to a rest-point. This requires cycling dynamics. This is why survival examples in  \cref{sec:paycomp} are more elaborated than the perhaps surprisingly simple examples of \cref{sec:simple}. Simpler examples of survival of dominated strategies under imitation dynamics based on payoff comparison may be given if we consider a population of players playing against an opponent with an exogeneously cycling behavior: see \cref{app:unilateral}.
\end{enumerate}

%\paragraph{From the replicator dynamics to the Smith dynamics}
\emph{From the replicator dynamics to the Smith dynamics. }
Consider again the protocol of \cref{ex:list} (making a list of strategies met), with a second step based on the proportional pairwise comparison rule,
$r_{ij} = [F_j - F_i]_+$. This revision protocol builds a bridge between the replicator dynamics and the Smith dynamics:
replicator dynamics are obtained for $m=1$ and the Smith dynamics (in the interior of the simplex) in the
limit $m \to +\infty$. This suggests that at least for this protocol and small values of $m$, survival of dominated
strategies will be more modest than with the Smith dynamics (lower domination level allowed, lower share of
the dominated strategy for a given domination level). This is what our preliminary numerical investigations also
suggest. A systematic investigation of these issues is left for future research.

%\paragraph{Favouring frequent strategies}
\emph{Favouring frequent strategies. }
On the other hand, imitation protocols favouring frequent strategies allow for survival of dominated strategies at very high frequencies, much higher that with the Smith dynamics or other standard innovative dynamics.
Conceptually, an advantage to frequent strategies could be given in innovative dynamics (i.e., such that strategies initially not played may appear), by assuming a form of risk-aversion of agents who would only be willing to adopt rare or unused strategies if the payoff of these rare strategies seem substantially higher than the payoff of better known strategies. For a risk-averse agent, this can be a rational attitude if information on the payoff of other strategies is noisy, with a greater variance for rare strategies, on which less information is available.

Note also that there is a certain degree of similarity between modifying a fair imitation protocol into one that benefits frequent strategies and adding to the payoffs of the game those of a pure coordination game.\footnote{In both cases, assume we start with twin strategies in the base game (before adding the coordination component), and most of the population playing the second strategy, and then add an increasingly high bonus to the first strategy, making the second one dominated. Initially, agents keep playing the second strategy due to either the advantage to frequent strategies or the added coordination component, but when the bonus becomes large enough, they switch to the first strategy. If the bonus for the first strategy is then reduced, and even made slightly negative, agents will keep playing the first strategy \textendash\ a hysteresis effect.}

%**********************************************************************
%***    APPENDICES
%**********************************************************************
\numberwithin{lemma}{section}		% for numbering  in the appendix
\numberwithin{corollary}{section}		% for numbering  in the appendix
\numberwithin{proposition}{section}		% for numbering  in the appendix
\numberwithin{equation}{section}		% for numbering in the appendix
\appendix

%----------------------------------------------------------------------
%%% APP: PROOFS
%----------------------------------------------------------------------
\section{Proofs of propositions on advantage to rare or frequent strategies}
\label{app:proofs}

In this section, the probability that a revising agent selects strategy $j$ at step 1 is independent of the revising agent's strategy, so we denote it by $p_j$ instead of $p_{ij}$.

\subsection{Meeting $m $ agents: Proof of \cref{prop:ex1}}

\begin{claim}
\label{cl:Bayes} It suffices to show that when $m$ is deterministic, then the first step is fair ($p_i=x_i$ for all $i$) for $m=1$ or $m=2$, and advantages rare strategies for any $m \geq 3$.
\end{claim}
This is a simple computation, which is left to the reader.
\begin{claim} The first step is fair for $m=1$ or  $m=2$\end{claim}
\begin{proof} This is obvious for $m=1$.
For $m=2$, this is because the selection steps boils down to selecting an agent uniformly at random, just breaking down the process in two stages:  first select two agents uniformly at random, then
among these two, select one of them, again uniformly.
\end{proof}

\begin{claim} For any fixed $m \geq 3$, the first step advantages rare strategies.
\end{claim}
\begin{proof} We divide the proof in four steps.

\vspace*{4pt}\noindent\textbf{Step 1.}
 Fix $m \geq 3$.
Let $0 \leq q \leq l \leq m$.
Let $E_{l, q}$ denote the event: among the $m$ agents met, $l$ play other strategies than $i$ or $j$ (so $\tilde{m}=m-l$ play $i$ or $j$) and these $l$ agents play $q$ different strategies.\footnote{Example: if $m= 5$, $i=1$, $j=4$, and the agents drawn are: one of type 1, two of type 2, two of type 3, then $l=4$ and $q=2$.} Then
\[\frac{p_i(x)}{x_i}= \sum_{ (q,l): 0 \leq q \leq l \leq m} P(E_{l,q}) \frac{P(i  | E_{l, q})}{x_i}\]

\vspace*{4pt}\noindent\textbf{Step 2.}  Now let $y_i=\frac{x_i}{x_i + x_j}$ and $y_j= 1-y_i$.
Condition on the event $E_{l,q}$.
If $l=m$, that is, if all $m$ agents met play strategies other than $i$ or $j$, then $P(i | E_{l,q})=0$.
Otherwise, each of the $\tilde{m}=m-l$ players playing $i$ or $j$ is of type $i$ with probability $y$ and the draws are independent.
So:

a) with probability $y_i^{\tilde{m}}$, all of these $\tilde{m}$ players are of type $i$; so there are exactly $q+1$ strategies encountered, including $i$ but excluding $j$.
Thus, $i$ is selected with probability $1/(q+1)$, and $j$ with probability $0$.

b) symmetrically, with probability $y_j^{\tilde{m}}$, all of the $\tilde{m}$ players are of type $j$, hence $i$ is selected with probability $0$  and $j$ with probability $1/(q+1)$

c) finally, with the remaining probability $1- y_i^{\tilde{m}} - y_j^{\tilde{m}}$, there are both players of type $i$ and players of type $j$ among these $\tilde{m}$ players, and both strategies are selected with probability $1/(q+2)$.

Summing up, if $l < m$, then
\begin{equation}
\label{eq:Elq}
P(i | E_{l, q})= \frac{1}{q+1} y_i^{\tilde{m}} + \frac{1}{q+2} \left(1- y_i^{\tilde{m}} - y_j^{\tilde{m}} \right)
\end{equation}

\vspace*{4pt}\noindent\textbf{Step 3.}
 Assume $m \geq 3$, $l \leq m-2$ (so $\tilde{m} \geq 2$), and $0 < x_i < x_j$.
Then \[\frac{P(i | E_{l,q})}{x_i} > \frac{P(j | E_{l, q})}{x_j}.\]

Let $A_i= (q+1)(q+2)P(i | E_{l,q}) / y_i$ and define $A_j$ similarly.
It suffices to show that $A_i > A_j$.
By \eqref{eq:Elq}:
\[y_i A_i=   (q+2)y_i^{\tilde{m}} + (q+1) (1- y_i^{\tilde{m}} - y_j^{\tilde{m}} )=  y_i^{\tilde{m}} + (q+1) (1- y_j^{\tilde{m}})\]
Noting that $\displaystyle 1- y_j^{\tilde{m}}= (1-y_j) \sum_{r=0}^{\tilde{m}-1} y_j^r = y_i  \sum_{r=0}^{\tilde{m}-1} y_j^r $ and dividing by $y_i$ we obtain:
\[A_i=   y_i^{\tilde{m}-1} + (q+1) \sum_{r=0}^{\tilde{m}-1} y_j^r= y_i^{\tilde{m}-1}  + (q+1) y_j^{\tilde{m}- 1} +  \sum_{r=0}^{\tilde{m}-2} y_j^r \]
and similarly for $A_j$.
It follows that  $A_i - A_j = T_1 + T_2$ with \[T_1= q (y_j^{\tilde{m}-1}  - y_i^{\tilde{m}- 1}) \mbox{ and } T_2 =  \sum_{r=0}^{\tilde{m}-2} (y_j^r - y_i^r).\] The term $T_1$ is always nonnegative and it is positive if $q \geq 1$, that is if  $l \geq 1$.
This is the case in particular if $l=m-2$ since $m \geq 3$.
The term $T_2$ is always nonnegative, and it is positive if $\tilde{m} \geq  3$, that is if $l \leq m-3$.
Since we assumed $l \leq m-2$, at least one of the terms $T_1$ and $T_2$ is positive.
Therefore, $T_1+ T_2 > 0$ and $A_i > A_j$.

\vspace*{4pt}\noindent\textbf{Step 4.}
 Assume $m \geq 3$ and $0 < x_i < x_j$.
Then $p_i/x_i > p_j/x_j$.

Indeed, it is easily seen that if $l=m$ or $l=m-1$, then $P(i | E_{l, q})/x_i= P(j | E_{l, q})/x_j$ (equal to $0$ if $l=m$, and to  $1/ [(x_i + x_j)(q+1)]$ if $l=m-1$).
Moreover, we just saw that if $l \leq m-2$, which happens with positive probability, then $P(i | E_{l, q})/x_i > P(j | E_{l, q})/x_j$.
Since
\[\frac{p_i}{x_i} = \sum_{0 \leq q \leq l \leq m} P(E_{l, q}) \frac{P(i | E_{l, q})}{x_i}.
\]
the result follows.
\end{proof}

\subsection{The majoritarian choice: Proof of \cref{prop:ex2}}

It suffices to show that if $m$ is deterministic, then step 1 is fair for $m=1$ or $m=2$, and advantages frequent strategies  for any $m \geq 3$.
The proof that step 1 is fair for $m=1$ or $m=2$ is as in \cref{prop:ex1}.
We now prove that for  $m \geq 3$, the first step favours frequent strategies.

Assume $x_i > x_j >0$ and let $y_i= x_i/(x_i + x_j)$ and $y_j =1-y_i$.
Consider a revising agent meeting $m\geq 3$ other agents.

\vspace*{4pt}\noindent\textbf{Case 1.}  Conditionally to the fact that only agents playing strategies $i$ and $j$ are met
(in a slight abuse of notation, we keep writing $p_i$ for the probability that $i$ is selected, without making clear in the notation that this is conditional on the fact that only agents playing $i$ or $j$ are met).

\vspace*{4pt}\noindent\textbf{Subcase 1.1
 (m odd, $m \geq 3$).}
  If $m= 2m'+1$, the probability that $i$ is selected is:
\[\frac{p_i}{y_i} = \frac{1}{y_i} \sum_{k= m'+1}^{m} \left(\begin{array}{c}  m' \\ k  \end{array}\right) y_i^k y_j^{m-k}= \sum_{k= m'+1}^{m} \left(\begin{array}{c}  m \\ k  \end{array}\right) y_i^{k-1} y_j^{m-k}\]
Similarly,
\[\frac{p_j}{y_j} =  \sum_{k= m'+1}^{m} \left(\begin{array}{c}  m \\ k  \end{array}\right) y_j^{k-1} y_i^{m-k}\]

Since for any $k \geq m'+1$, we have $k-1 \geq m' \geq m - (m'+1) \geq m-k$, it follows that the first expression is term by term greater than the second one, and strictly greater for all terms $k > m'+1$.
Such terms exists because $m= 2m'+1 \geq 3$ implies $m > m'+1$.
It follows that $p_i/y_i > p_j/y_j$.

\vspace*{4pt}\noindent\textbf{Subcase 1.2 (m even, $m\geq 4$).} If $m=2m'$, then there may be a tie, if both strategies are met $m'$ times, in which case they are selected with probability $1/2$.
Thus we get:
\begin{equation}
\label{eq:app1prot2}  \frac{p_i}{y_i} = \frac{1}{2} \left(\begin{array}{c}  m \\ m'  \end{array}\right) y_i^{m'-1} y_j^{m'}
+ \sum_{k= m'+1}^{m} \left(\begin{array}{c}  m \\ k  \end{array}\right) y_i^{k-1} y_j^{m-k}
\end{equation}
Note that if $k \geq m'+1$, then
\[y_i^{k-1}y_j^{m-k} \geq y_i^{m'} y_j^{m- (m'+1)}= y_i^{m'} y_j^{m'-1}.\]
Moreover, the inequality is strict for any $k \geq m'+2$, in particular for $k=m$, since we assumed $m=2m' \geq 4$.
Thus, factorizing by $y_i^{m'-1}y_j^{m'-1}$, we obtain:
\[\frac{p_i}{y_i} > y_i^{m'-1}y_j^{m'-1} \left[  \frac{1}{2} \left(\begin{array}{c}  m \\ m'  \end{array}\right) y_j
+ \sum_{k= m'+1}^{m} \left(\begin{array}{c}  m \\ k  \end{array}\right) y_i \right] \]

A similar (but reverse) inequality holds for $p_j/y_j$.
Using both inequalities, we obtain:
\[\frac{p_i}{y_i} - \frac{p_j}{y_j} > y_i^{m'-1}y_j^{m'-1} (y_i - y_j) \left[\sum_{k= m'+1}^{m} \left(\begin{array}{c}  m \\ k  \end{array}\right)  - \frac{1}{2} \left(\begin{array}{c}  m \\ m'  \end{array}\right) \right]\]

We let the reader check that the first term in the summation suffices to show that the bracket is nonnegative, so that $p_i/y_i > p_j/y_j$.

\vspace*{4pt}\noindent\textbf{Case 2.} Now consider the general case.
Out of the $m$ players met, let $m_k$ denote the number of players playing strategy $k$.
Let $E(l, b, q)$ denote the event: out of the $m$ players met, $l= \sum_{k \notin \{i, j\}} m_k$ play strategies different from $i$ and $j$, $b= \max_{k \notin \{i,j\}} m_k$ is the highest number of occurence of a strategy different from $i$ and $j$, and there are $q$ strategies $k \notin \{i, j\}$ such that $m_k=b$. Condition on this event.
Again, we write $p_i$ instead of $P(i |E(l, b, q))$.

We dealt with the case $l=0$ in Case 1, so we may assume $l \geq 1$ hence $b \geq 1$ and $q \geq 1$.
Let $\tilde{m} = m - l$ be the number of agents met playing $i$ or $j$.

\noindent\textbf{Subcase 2.1. $b > \tilde{m}$.} Then $i$ and $j$ cannot be selected, hence $p_i=p_j=0$.

\noindent\textbf{Subcase 2.2. $\tilde{m} \geq 2b+1$.} Then one of the strategies $i$ and $j$ will win for sure.
Moreover, $\tilde{m}  \geq 3$, and the proof is as in Case 1, replacing $m$ with $\tilde{m}$.

\vspace*{4pt}\noindent\textbf{{Subcase 2.3. $\tilde{m}= 2b$.}} This is similar to Subcase 1.2., replacing $m$ with $\tilde{m}$, with the twist that if $m_i=m_j=b$, the strategies $i$ and $j$ are not selected with probability $1/2$, but $1/(q+2)$. The factor $1/2$ in Eq. \eqref{eq:app1prot2} thus becomes $1/(q+2)$. Since $q \geq 1$, it is then easy to check that $p_i/y_i > p_j/y_j$ even if $\tilde{m}=2$ (while we had to require $m \geq 4$ in Subcase 1.2).

\vspace*{4pt}\noindent\textbf{{Subcase 2.4. $b \leq \tilde{m} \leq 2b - 1$.}} This case is similar to Subcase 1.1.
We get:
\[\frac{p_i}{y_i} = \frac{1}{q+1} \left(\begin{array}{c}  \tilde{m}  \\ b  \end{array}\right) y_i^{b-1} y_j^{\tilde{m} - b}
+ \sum_{k= b+1}^{\tilde{m}} \left(\begin{array}{c}  \tilde{m} \\ k  \end{array}\right) y_i^{k-1} y_j^{\tilde{m}-k}\]
and a symmetric expression for $p_j/y_j$.
Because $\tilde{m} \leq 2b - 1 \Rightarrow b-1 \geq \tilde{m} - b$,
it follows that the expression for $p_i/y_i$ is term by term greater than the expression for $p_j/y_j$, with a strict inequality for the term $k=\tilde{m}$, unless $\tilde{m}=1$.
It follows that if $\tilde{m}=1$, $p_i/y_i= p_j/y_j$, and if $\tilde{m} > 1$, then $p_i/y_i > p_j/y_j$.

\emph{To conclude:} for any $l$, $b$, $q$, $P(i | E(l,b,q))/y_i \geq P(j | E(l, b, q))/y_j$, with a strict inequality in some cases occurring with positive probability.
Since $$p_i/y_i= \sum_{l, b, q} P(E(l, b, q)) P(i | E(l,b,q))/y_i,$$
%highlighted because otherwise the line is too wide
 it follows that $p_i/y_i > p_j/y_j$.

%----------------------------------------------------------------------
%%% APP: IMITATION
%----------------------------------------------------------------------
\section{Imitation protocols not of the form \eqref{eq:gen2step}}
\label{app:moreprot}

We note here that our results would also apply to protocols that cannot be neatly separated in two steps in the sense of Eq. \eqref{eq:gen2step}. Reconsider \cref{ex:list} from \cref{sec:ImProc}, where a revising agent meets several other agents and makes a list of the strategies they play. We assumed then that he would investigate just one of these strategies. Instead, the revising agent could obtain information on the payoffs of all those strategies.
This makes sense if getting information on payoffs of strategies met is cheap.
In our concrete example, after meeting strategies $1$, $2$, $3$, the revising agent would obtain information on the payoffs $F_1$, $F_2$, $F_3$, and adopt one of these strategies with a probability that depends on all these payoffs, and possibly his own.
For instance, he could adopt strategy $j \in \{1, 2, 3\}$ with probability $f(F_j)/ (1 + \sum_{k=1, 2, 3} f(F_k))$ with $f$ positive increasing, or with probability $[F_j - F_i]_+/(1 + \sum_{k=1, 2, 3} [F_k - F_i]_+)$.
Such protocols cannot easily be put in the form \eqref{eq:gen2step}.
Nevertheless, the resulting dynamics still favour rare strategies in the sense that when two strategies have the same payoff, the rarest one has a higher per-capita growth rate; thus, as long as the switching rates $\rho_{ij}$ are regular enough in $(F, x)$, versions of our results would apply. However, our results do not apply to \emph{discontinuous} imitative variants of the best-reply dynamics, such as imitating a best-reply to the current population state among the strategies met.

%----------------------------------------------------------------------
%%% APP: UNILATERAL
%----------------------------------------------------------------------
\section{Unilateral approach: Simple examples for comparison based imitation processes}
\label{app:unilateral}

In this section, we adopt a unilateral approach, in the spirit of (Viossat, 2015 \cite{11}).
That is, we study the evolution of behavior in a large population of players (the focal population, player 1) facing an unknown opponent (the environment, player 2), whose behavior we freely choose.
This allows to provide simple examples of survival of dominated strategies even for dynamics based on payoff comparison.

Specifically, let us denote by $G_{\ep}$ a $3 \times 2$ game where the payoffs in the focal population are as follows:
\begin{equation}
\label{eq:GameUni}
\begin{array}{cc}
  & \begin{array}{cc}
    L \hspace{0.2 cm} & \hspace{0.2 cm} R \\
    \end{array} \\
\begin{array}{c}
 1 \\
 2  \\
 3 \\
\end{array}
& \left(\begin{array}{cc}
         1       &       0        \\
 	0	& 	1	 \\
         -\ep &  1- \ep  \\
\end{array}\right)\\
\end{array}
\end{equation}

As before, $x_i(t)$ denotes the frequency of strategy $i \in \{1, 2,3 \}$ in the focal population.
We make the following assumptions:

(A1) For $i=1, 2, 3$, when the opponent plays $Y \in \{L, R\}$, then $$\dot x_i = x_i g^Y_i(x)$$ for some growth-rate function $g_i^Y : X \to \mathbb{R}$ that is Lipschitz continuous in $x$
and depends continuously on the parameter $\ep$ (here, $X$ denotes the simplex of possible population states for the focal population).

(A2) When $\ep=0$, if $x_1 \notin \{0, 1\}$,
then $g_1^L(x) >0$ and $g_1^R(x) < 0$.\\
We also assume that at least one of the conditions (A3), (A3') below holds:

(A3)  When $\ep=0$, if $x_3 < x_2$, then $g_3^L (x) \geq g_2^L(x)$ and  $g_3^R (x) > g_2^R(x)$\\
or

(A3') When $\ep=0$, if $x_3 < x_2$, then $g_3^L (x) > g_2^L(x)$ and  $g_3^R (x) \geq g_2^R(x)$\\

Assumption (A1) is a regularity assumption.
Assumption (A2) is weaker than Positive Correlation.
Assumption (A3) or (A3')
 is a form of advantage to rare strategies.
These assumptions are satisfied, for instance, by any dynamics arising from a revision
 protocol of form \eqref{eq:gen2step} with $\lambda_{ij}$, $r_{ij}$ Lipschitz continuous in $x$ and continuous in $F$,
 $r_{ij}$ with the sign of $[F_j - F_i]_+$,
 and favouring rare strategies in the sense of \cref{def:adv}.

\begin{proposition}
\label{prop:app3}
Fix $\eta > 0$.
Let $\delta$, $x_{\min}$, $x_{\max}$ be real numbers such that $0 < \delta < x_{\min} < x_{\max} < 1- \delta$.
Let $K_{\delta}= \{x \in X | \min(x_1, 1-x_1) \geq \delta\}$.
 Assume that the opponent plays $L$ until the first time $\tau$
such that $x_1(\tau) \geq x_{\max}$, then plays $R$ for $t > \tau$ until $x_1 = x_{\min}$, then plays $L$ again until $x_1= x_{\max}$, etc.\footnote{The fact that the opponent plays a
discontinuous strategy simplifies the exposition but could be replaced by a similar behavior with smooth transitions. Due to this discontinuity, the frequencies $x_i(t)$ are only piecewise $C^1$,
but it may be shown that this creates no technical difficulty.}
Then there exists $\bar{\ep}>0$ such that for any $\ep \in [0, \bar{\ep}]$ and any initial condition
$x(0) \in K_{\delta} \cap \mathrm{int}(X)$, $\liminf x_3(t) > (1- x_{\max}) \left(\frac{1}{2} - \eta \right)$.
\end{proposition}
\begin{proof}
The intuition is that when $\ep=0$, the shares of strategies 2 and 3 tend to become equal. Thus, $\liminf x_3(t) = (1-\limsup x_1)/2= (1-x_{\max})/2$.
We then need to show that for a sufficiently small perturbation of payoffs, $\liminf x_3$ remains close to $(1-x_{\max})/2$. By contrast with \cref{th:hypno},
we do not deal with an autonomous system of differential equations, but with a controlled system. This is why the proof below does not rely on continuity of attractors but on a direct analysis.

To fix ideas, assume that $(A3)$ holds. The proof when $(A3')$ holds is similar.  Throughout, we assume that $x(0) \in K_{\delta} \cap \mathrm{int}(X)$.
By (A1),  (A2) and compactness of $K_{\delta}$, there exist positive real numbers $\bar{\ep}$, $\alpha_1$, $\alpha_2$
such that, for any $\ep$ in $[0, \bar{\ep}]$ and any $x \in K_{\delta}$, $\alpha_1 \leq \dot{x}_1 \leq \alpha_2$ when the
opponent plays $L$ and $- \alpha_2 \leq \dot{x}_1 \leq - \alpha_1$ when she plays $R$.
It follows that $x(t)$ eventually
enters the compact set $$K = \{x \in X, x_{\min} \leq x_1 \leq x_{\max}\},$$ and never leaves, oscillating between $x_{\min}$ and $x_{\max}$.
Moreover, the time to travel from the hyperplane $x_1 = x_{\min}$ to the hyperplane $x_1= x_{\max}$ (or back) is always
between $$T_{\min}= \frac{x_{\max}- x_{\min}}{\alpha_2} \text{ and } T_{\max}= \frac{x_{\max} - x_{\min}}{\alpha_1}.$$
Note that $\liminf(x_2 + x_3) = 1 - x_{\max}$.
Thus  if suffices to show that, possibly up to lowering $\bar{\ep}$,
$$\liminf \frac{x_3}{x_2 + x_3} \geq \frac{1}{2} - \eta.$$ We first show that $\limsup \frac{x_3}{x_2 + x_3}  \geq \frac{1-\eta}{2}$.

Assume by contradiction that this is not the case.
Then from some time $T$ on, $$x(t) \in \tilde{K} = \left\{ x \in K, \frac{x_3}{x_2 + x_3}  \leq \frac{1-\eta}{2} \right\}.$$
By (A1), (A3) and compactness of $\tilde{K}$, and up to lowering $\bar{\ep}$, we may assume that there exist positive real
numbers $\beta_1$ and $\beta_2(\ep)$ such that for any $x \in \tilde{K}$ and any $\ep \in [0, \bar{\ep}]$,
\begin{equation}
\label{eq:compgrowth}
 g_3^R(x) - g_2^R(x) \geq \beta_1 \text{ and } g_3^L(x) - g_2^L(x) \geq - \beta_2(\ep)
 \end{equation}
with $\beta_1$ independent of $\ep$ and $\beta_2(\ep) \to 0$ as $ \ep \to 0$.
Up to lowering $\bar{\ep}$ again, we may
assume that $$C:= \beta_1 T_{\min} - \beta_2(\ep) T_{\max} >0.$$ Now let $t_{2k}$ and $t_{2k+1}$ be  the $k^{th}$ time greater
than $T$ such that $x_1 = x_{\min}$ and $x_1 = x_{\max}$, respectively.
Note that
$\frac{d}{dt} \ln(x_3/x_2) = g_3^Y (x) - g_2^Y(x)$ when the opponent plays $Y$.
Integrating between $t_{2k}$ and $t_{2k+2}$ and using
\eqref{eq:compgrowth} we obtain that between $t_{2k}$ and $t_{2k+2}$, $\ln(x_3/x_2)$ increases by at least $C$.
Since $C>0$,
this implies that $x_3/x_2 \to +\infty$, a contradiction.
Therefore,
$$\limsup_{t \to +\infty} \frac{x_3}{x_2 + x_3}(t) \geq \frac{1-\eta}{2}.$$

Moreover, since $\beta_2(\ep) \to 0$ as $\ep \to 0$, up to lowering $\bar{\ep}$ again, we may assume that between $t_{2k}$ and $t_{2k+1}$,
$x_2/(x_2+ x_3)$ does not decrease by more than $\eta/2$.
It may be shown that this ensures that
$\liminf \frac{x_2}{x_2 + x_3}   \geq \frac{1-\eta}{2} - \frac{\eta}{2} = \frac{1}{2} - \eta$.
This concludes the proof.
\end{proof}
Note that for $x_{\max}$ and $\eta$ small enough, $\liminf x_3$ may be made arbitrarily close from $1/2$.
If we replace Assumptions (A3), (A3') by the same assumptions but when $x_3 < x_2$, thus giving an advantage to frequent strategies, then we obtain that for $\ep$ small enough and an open set of initial conditions, $\liminf x_3$ may be made arbitrarily close to $1$.

\begin{figure}[t]
  \centering
%  \vspace{0.5cm}
    \includegraphics[width=0.45 \textwidth]{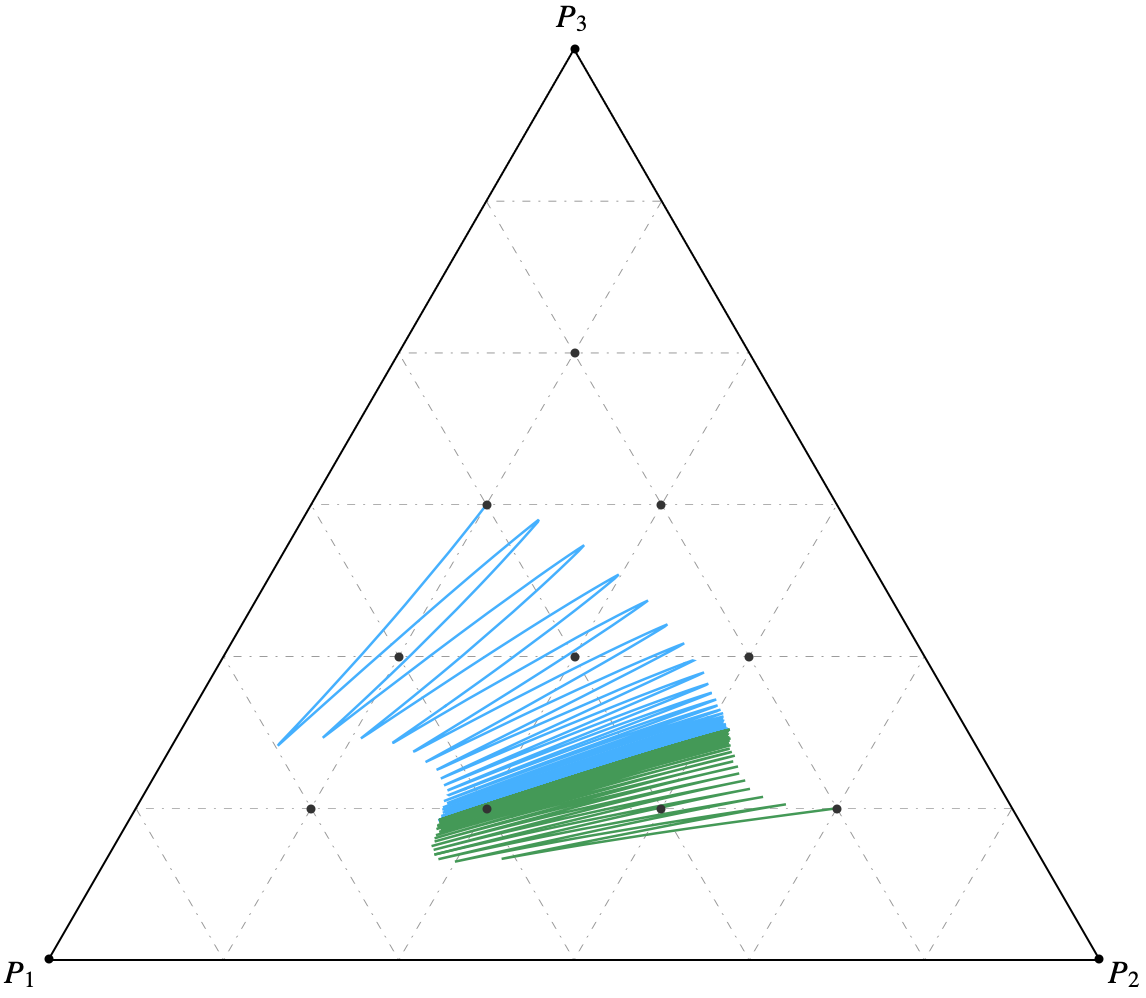}   \quad \includegraphics[width=0.45 \textwidth]{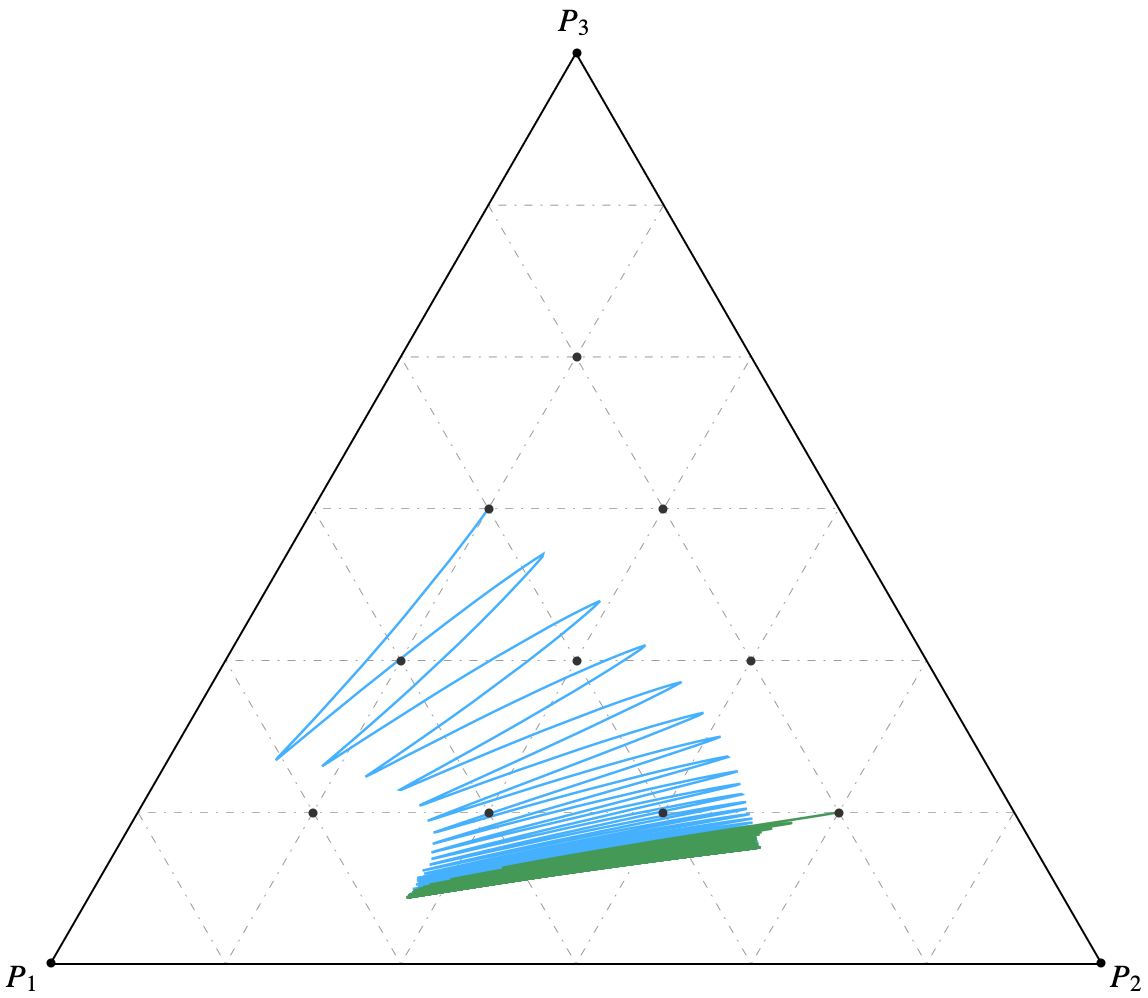}\\
    \bigskip
          \includegraphics[width=0.45 \textwidth]{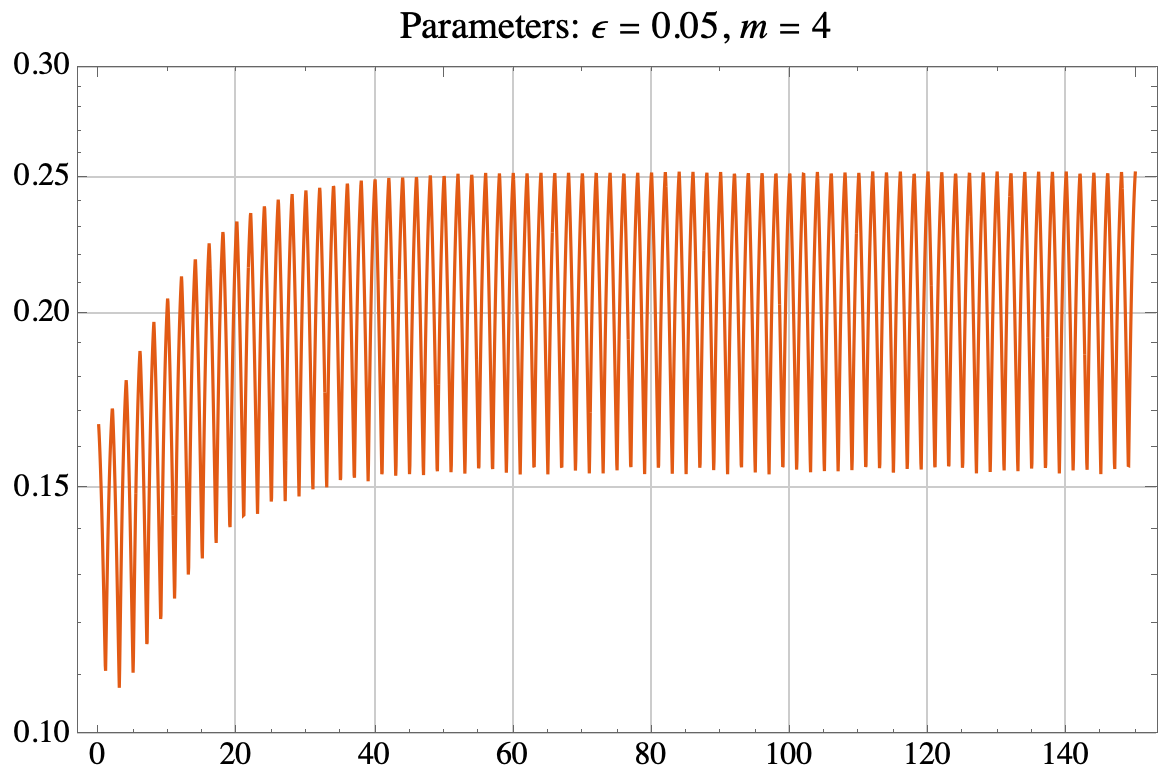} \quad  \includegraphics[width=0.45 \textwidth]{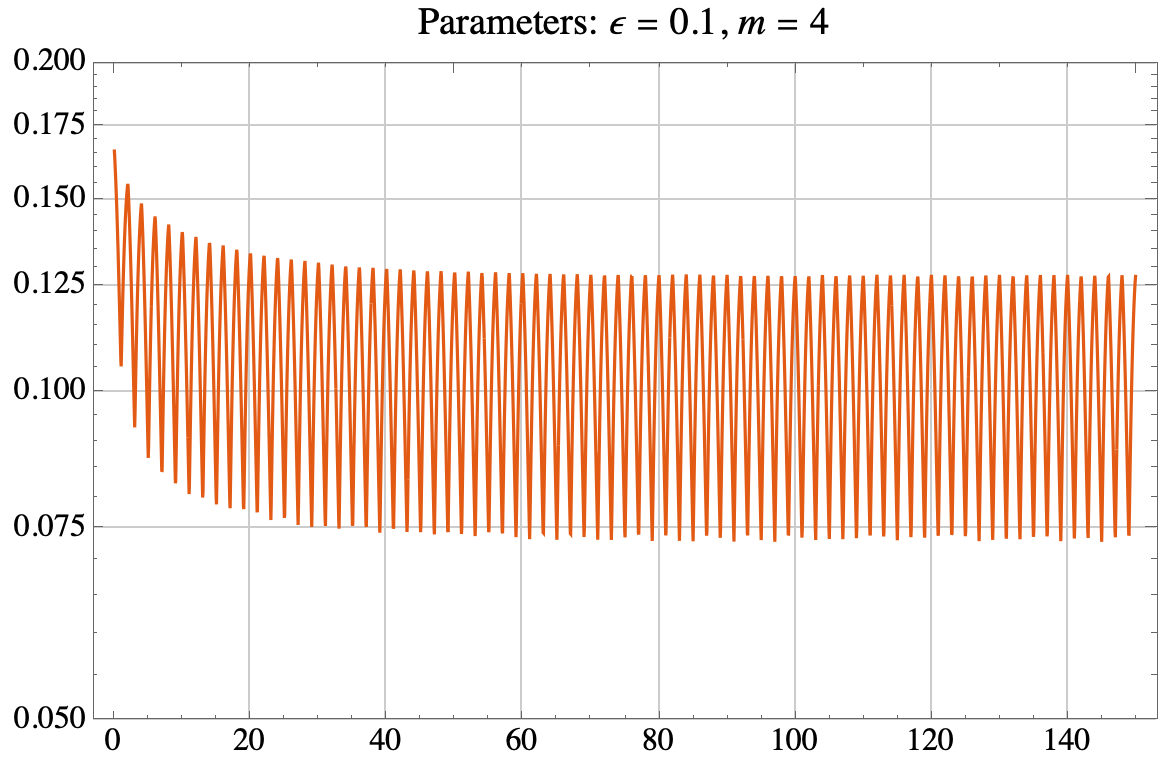}
%\vspace{0.5cm}
\caption{\textbf{Imitation dynamics favouring rare strategies in Game \eqref{eq:GameUni} against an oscillating behavior of the opponent}. Top-panels: the solution $x(t)$. The point $P_i$ corresponds to the population state where everybody plays strategy $i$. Bottom-panels: frequency of the strictly dominated strategy (strategy 3). Left-column: $\ep = 0.05$. Right-column: $\ep= 0.1$. In blue and green, two solutions with respective initial conditions $(1/3, 1/6, 1/2)$ and $(1/6, 2/3, 1/6)$.
}
  \label{FigB}
\end{figure}

\cref{FigB} depicts imitation dynamics with payoffs in the focal population described by the payoff matrix \eqref{eq:GameUni} and a periodic behavior of Player 2 that smoothly approximates playing $L$ on time-intervals of the form $[2k, 2k+1)$ and $R$ on time-intervals of the form $[2k+1, 2k+2)$, where $k$ is an integer (at time $t$, Player 2 puts probability $y(t) = \frac{1+ \sin^{1/9} (\pi t)}{2}$ on strategy L). As in \cref{FigC}, the dynamics of the focal population are derived from a two-step protocol of form \eqref{eq:gen2step}, with a first step as in \cref{ex:other}
(trying to meet an agent playing another strategy),
with $m= 4$, and a second step based on payoff comparison $r_{ij} = [F_j - F_i]_+$. \cref{FigB} illustrates that survival of dominated strategies can also occur if the behavior of the opponent is smooth and independent of the current population state in the focal population. The average frequency of the dominated strategy is around $20\%$ with a domination margin of $\ep = 0.05$, and around $10\%$ with a domination margin of $\ep = 0.1$.

For an advantage to frequent strategies, survival of the dominated strategy in this example seems less robust:  if the behavior of the opponent oscillates in a way that is independent of the population state in the focal population, what happens in most simulations is that initially either strategy 1 or strategy 3 takes over, as deviations from an approximately equal share of these strategies get amplified by the advantage to frequent strategies. In the first case, the solution converges to the mixed strategy putting probability 1 on the first strategy. In the second case, strategy 1 gets extinct, and then, since the second step of the protocol is based on payoff comparison, strategy 2 drives strategy 3 extinct.

%**********************************************************************
%***    ACKNOWLEDGMENTS
%**********************************************************************
\section*{Acknowledgments}
\begingroup
\small
The first author is grateful for financial support by
the French National Research Agency (ANR) in the framework of
the ``Investissements d'avenir'' program (ANR-15-IDEX-02),
the LabEx PERSYVAL (ANR-11-LABX-0025-01),
MIAI@Grenoble Alpes (ANR-19-P3IA-0003),
and the bilateral ANR-NRF grant ALIAS (ANR-19-CE48-0018-01).
\endgroup

%**********************************************************************
%***    BIBLIOGRAPHY
%**********************************************************************

\end{document}